\documentclass[11pt]{article}

\usepackage{amsmath,amsfonts,amssymb,amscd,latexsym,mathrsfs,amsthm}

\usepackage{color}

\newtheorem{theorem}{Theorem}[section]
\newtheorem{lemma}[theorem]{Lemma}
\newtheorem{proposition}[theorem]{Proposition}
\newtheorem{corollary}[theorem]{Corollary}

\newtheorem{remark}[theorem]{Remark}
\newtheorem{assumption}[theorem]{Assumption}

\newcommand{\C}{\ensuremath{\mathbb{C}}}

\newcommand{\R}{\ensuremath{\mathbb{R}}}
\renewcommand{\S}{\mathbb{S}}
\def\V{\mathcal V}

\def\FF{\mathcal F}
\def\H{\mathcal H}
\def\d{\mathrm{d}}
\def\SS{\mathcal S}
\newcommand{\K}{\mathcal K}
\newcommand{\B}{\mathcal B}

\newcommand{\D}{\mathcal D}

\def\ci{\mathrm{ci}}
\def\si{\mathrm{si}}

\def\G{\mathcal G}
\newcommand{\T}{\mathcal T}
\def\Ker{\mathrm{Ker}}
\def\Nul{\mathrm{Nul}}
\def\HS{\mathfrak h}
\def\Hrond{\mathscr H}
\newcommand\cf{{\em cf.}}

\begin{document}

\title{
Low energy spectral and scattering theory \\
for relativistic Schr\"odinger operators
}

\author{Serge Richard$\,^1$\footnote{On leave from Universit\'e de Lyon; Universit\'e
Lyon 1; CNRS, UMR5208, Institut Camille Jordan, 43 blvd du 11 novembre 1918, F-69622
Villeurbanne-Cedex, France.
This work was done during a period of support by the Japan Society for the Promotion of Science and by
``Grants-in-Aid for scientific Research''.}~~and Tomio Umeda$\,^2$\footnote{Supported by
 the Japan Society for the Promotion of Science
     ``Grant-in-Aid for Scientific Research'' (C)
    No.  21540193.}}
\date{}
\maketitle

\vspace{-8mm}

\begin{quote}
\begin{itemize}
\item[$^1$] Graduate school of mathematics,
University of Nagoya,
Chikusa-ku, \\
Nagoya 464-8602,
Japan \\
E-mail: {\tt richard@math.nagoya-u.ac.jp}
\item[$^2$] Department of Mathematical Sciences, University of Hyogo, Shosha,\\ Himeji 671-2201, Japan \\
E-mail: {\tt umeda@sci.u-hyogo.ac.jp}
\end{itemize}
\end{quote}

\begin{abstract}
Spectral and scattering theory at low energy for the relativistic Schr\"odinger operator are investigated. Some striking properties at thresholds of this operator are exhibited, as for example the absence of $0$-energy resonance. Low energy behavior of the wave operators and of the scattering operator are studied, and stationary expressions in terms of generalized eigenfunctions are proved for the former operators. Under slightly stronger conditions on the perturbation the absolute continuity of the spectrum on the positive semi axis is demonstrated. Finally, an explicit formula for the action of the free evolution group is derived. Such a formula, which is well known in the usual Schr\"odinger case, was apparently not available in the relativistic setting.
\end{abstract}

\textbf{2000 Mathematics Subject Classification:} 81U05, 35Q40, 47F05

\smallskip

\textbf{Keywords:} relativistic Schr\"odinger operators, low energy, scattering theory, wave operators, dilation group


\section{Introduction}\label{sec3}
\setcounter{equation}{0}

The aim of this paper is to study
the spectral and scattering theory of the operator $H=\sqrt{-\Delta}+V$ in
$L^2(\R^3)$
 with a special emphasize on low but positive energies.
Various properties of this so-called relativistic Schr\"odinger operator have already been exhibited in \cite{BN,U0,U, W}, but its corresponding wave operators and scattering operator still deserved investigations. Obviously, the natural comparison operator is the free operator $H_0:=\sqrt{-\Delta}$, while for the perturbation it will be assumed that $V$ is a measurable real function on $\R^3$ satisfying
\begin{equation}\label{condition1}
|V(x)| \leq {\rm Const.}~\langle x \rangle^{-\sigma}
\end{equation}
for some $\sigma>1$ and almost every $x \in \R^3$. Here, we have used the standard notation $\langle x \rangle:= (1+x^2)^{1/2}$.

Now, note that similar investigations for the scattering theory in the usual Schr\"odinger case ({\it i.e.}~for the operator $-\Delta + V$) are part of a piece of folklore. Indeed, based on the seminal work \cite{JK}, the low energy behavior of the wave operators and of the scattering operator can be derived from stationary expressions for these operators.
As for the relativistic Schr\"odinger operator,
on the other hand, the absence of existing information on the behavior of $(H-\lambda \mp i 0)^{-1}$ as $\lambda \searrow 0$ prevented such a study. For that reason, part of the present work is dedicated to the study of various properties at low energy of the resolvent of the free operator as well as of the perturbed operator. Only once these preliminary results are obtained, further investigations on the scattering theory can be performed.

So, let us be more precise about the framework and about the results. By assuming that $V$ satisfies Condition \eqref{condition1}, then both $H_0$ and $H$ are self-adjoint operators with domain equal to the Sobolev space of order $1$ on $\R^3$. In addition, the spectrum of $H_0$ consists only of an absolutely continuous part on $[0,\infty)$, while $H$ possesses absolutely continuous spectrum on $[0,\infty)$ together with a possible discrete set of eigenvalues on $\R$ which can accumulate only at $0$ or at $\infty$. These results follow from limiting absorption principles which have already been derived in \cite{BN}.

Now, our first task is the study of the $0$-energy threshold. In particular, one shows that in suitable spaces the operator $(H_0-\lambda\mp i 0)^{-1}$ admits an explicit limit as $\lambda \searrow 0$. Then, one proves that $0$ is generically not an eigenvalue for $H$, and that this operator does not possess $0$-energy resonance, see the discussion following Lemma \ref{lem:0eigenvl-1} for a precise statement.
In the same vein, one also shows that if $0$ is not an eigenvalue of $H$, then $0$ cannot be an accumulation point of positive eigenvalues of $H$.
One should note that such a property has no analog for usual Schr\"odinger operators.
These various spectral results are all derived in Section \ref{sec0behavior}.

Our next task is the derivation of a particular stationary expression for the wave operators $W_\pm$;
the definition of $W_\pm$ can be found at the beginning of Section \ref{secsta}. In fact, such a formula was already announced in \cite{U0} but the full proof was lacking. The construction is based on generalized eigenfunctions which can be proved to exist if $V$ satisfies Condition \eqref{condition1} for $\sigma>2$. The entire Section \ref{secsta} is devoted to this proof and the main result expressing the wave operators in terms of generalized eigenfunctions is contained in Proposition \ref{proposition_stationary}.

Section \ref{secwave} contains our main new results on the wave operators. Obviously, since $W_\pm$ can not be diagonalized in
the spectral representation of $H_0$ or of $H$, studying the low energy behavior of $W_\pm$ has to be suitably defined. In fact, our approach relies on the use of the unitary dilation group, which has often been at the root of investigations on rescaled Schr\"odinger operators, see for example \cite{AGHH}.
So, let us recall the action of the dilation group $\{U_\tau\}_{\tau \in \R}$ on any $f \in L^2(\R^3)$,
namely $[U_\tau f](x) = e^{3\tau/2}f(e^\tau x)$ for any $x \in \R^3$. Then, the following two relations are of importance, namely $U_{-\tau}H_0U_\tau = e^\tau H_0 $ and
\begin{equation}\label{firsttime}
U_{-\tau}\ W_\pm(H_0+V,H_0)\ U_\tau = W_\pm(H_0+e^{-\tau}V_\tau,H_0),
\end{equation}
where $V_\tau(x) = V(e^{-\tau}x)$ for all $x \in \R^3$. Note that for clarity, the dependence of $W_\pm$ on both self-adjoint operators used to define them is mentioned. In that setting, our investigations are concentrating on the behavior of the r.h.s.~term of \eqref{firsttime} as $\tau \to -\infty$. As we shall see in Section \ref{secS}, this study has a direct consequence on the behavior of the scattering operator at low energy, which is well defined since the scattering operator is diagonal in the spectral representation of $H_0$.

Now, as already mentioned above, asymptotic properties of $W_\pm$ can only be derived once suitable information on the resolvent of $H$ are obtained. For that purpose, we provide a rather detailed analysis of the operator $\big(1+u(H_0-\lambda\mp i 0)^{-1} v\big)^{-1}$, with $v = |V|^{1/2}$ and $u = |V|^{1/2}{\rm sgn}(V)$, as $\lambda \searrow 0$, see Proposition \ref{lembehavior0}
where $\big(1+u(H_0-\lambda\mp i 0)^{-1} v\big)^{-1}$ is denoted by $B(\lambda \pm i 0)$.
Note that our analysis holds if $0$ is not an accumulation point of positive eigenvalues. A comment on this implicit assumption is formulated below. Then, with this information at hand, the main result of Section \ref{secwave} states that the strong limit $s-\lim_{\tau \to -\infty} U_{-\tau} W_\pm(H,H_0)U_\tau$ is equal to $1$.

The main consequence of this statement concerns the low energy behavior of the scattering operator $S$ defined by $W_+^*W_-$. In that setting, this corollary states that $s-\lim_{\tau \to -\infty} U_{-\tau} S U_\tau =1$.
Additionally, one also proves a uniform convergence
of the scattering operator
in the spectral representation of $H_0$, namely $u-\lim_{\lambda \searrow 0} S(\lambda) = 1$,
where $S(\lambda)$ is the scattering matrix.
This result indicates that there is a significant difference
between usual
Schr\"odinger operators and relativistic Schr\"odinger operators
in terms of the low energy asymptotics of the scattering matrices: compare
the result of the present paper with the corresponding ones of
\cite{DS1,DS2,F,JK}. What causes
this difference is the absence of $0$-energy resonances for
relativistic Schr\"odinger operators.
These statements and their proofs correspond to the content of Section \ref{secS}.

Now, the non-existence of embedded eigenvalues should certainly deserve more attention for the present model. However, since investigations on this question for Schr\"odinger operators always involve a rather heavy machinery, we do not expect that this question can be easily solved for the present relativistic model. On the other hand, by assuming stronger conditions on $V$, one can rather easily deduce from an abstract argument that the spectrum of $H$ on $[0,\infty)$ is purely absolutely continuous. Section \ref{sectionac} is devoted to such a result.
We clearly suspect that the assumptions on $V$ are much too strong for the non-existence of positive eigenvalues, but since the argument is rather simple we have decided
to present it for completeness.
The proof is based on an abstract result obtained in \cite{Ri06}

Finally, in an appendix, we derive an explicit formula for the action of
the unitary propagator $e^{-itH_0}$.
Such a formula, which is well known in the Schr\"odinger case,
was apparently not known in the relativistic case.

In summary, this work contains various results on
the low energy behavior
of the spectral and the scattering theory
of relativistic Schr\"odinger operators.
A similar study for the high energy behavior of these operators would certainly be valuable, and accordingly, a better understanding of the existence or the absence of positive eigenvalues should also deserve some attention. Only once these pre-requisites are fulfilled, a rather complete picture of the scattering theory for relativistic Schr\"odinger operators would be at hand.
Note finally that we have confined our attention to the three dimensional case, although it is apparent that some results of the present paper could be generalized to higher dimensional cases (see for example  \cite{Wei} for a detailed study of the completeness of the generalized eigenfunctions for odd dimensional relativistic Schr\"odinger operators).
However, it is well-known that for non-relativistic Schr\"odinger operators, the $0$-energy behavior of the resolvent highly depends on the space dimension, and we expect that a similar phenomenon also takes place for its relativistic counterpart.

\paragraph{Notations:}
We introduce the notations which will be used in the present paper.

We shall mainly work in the Hilbert space $\H:=L^2(\R^3)$ with norm and scalar product denoted by $\|\cdot\|_\H$ and $\langle \cdot,\cdot\rangle_\H$. Our convention is that the scalar product is linear in its first argument. The weighted Sobolev spaces of order
$t\in \R$ and weight $s\in \R$
are denoted by $\H^t_s$.
Note that if $s$ or $t$ is equal to $0$, we simply omit it. A norm on $\H_s^t$ is provided by the expression
$$
\big\|f\big\|_{\H_s^t} = \big\|\langle X\rangle^s \langle D\rangle^t f \big\|_\H,
$$
where $X$ is the position operator and $D=-i\nabla$ is its conjugate operator in $\H$. With these notations, the usual Laplace operator $-\Delta$ is equal to $D^2$.

The notation $C_0(\R^3)$ denotes the set of continuous functions on $\R^3$ which vanish at infinity. The Schwartz space on $\R^3$ is denoted by $\SS(\R^3)$ while $C_c^\infty(\R^3)$ defines the set of smooth functions on $\R^3$ with compact support.

By extension, for any $s \in \R$ we denote by  $\langle  \cdot, \cdot \rangle_{s,-s}$ the pairing between
$\H_s$ and $\H_{-s}$, namely for $f\in \H_s$ and $g \in \H_{-s}$:
\begin{equation*}
\langle f, g \rangle_{s, -s} = \int f(x) \overline{g(x)} \,\d x\ .
\end{equation*}
If $f$ belongs to $\SS(\R^3)$ and $g$ is a tempered distribution, we shall use the notation $\langle f,g\rangle_{\SS,\SS'}$ for their pairing.
Similarly, if $f\in\H_s^{-t}$ and $g\in\H_{-s}^{t}$, we shall use the notation
$\langle f,g\rangle_{\H_s^{-t},\H_{-s}^{t}}$.
The usual Fourier transform of $f$ is denoted both by $\hat f$ and $\FF f$ and is defined explicitly on any $f\in \SS(\R^3)$ by
\begin{equation*}
[\FF f](k)= (2\pi)^{-3/2} \int_{\R^3}f(x)\;\!e^{-ix\cdot k}\;\!\d x\ .
\end{equation*}
The same notation is used for its standard extension to tempered distributions. As well known, this map is a unitary operator in $\H$, and its inverse is denoted by $\FF^*$.

For a pair of Hilbert spaces $\G$ and $\H$, $\B(\G;\H)$ denotes the Banach space of all bounded and linear operators from $\G$ to $\H$, and $\K(\G;\H)$ the subset of compact operators. We set $\B(\H)$ for $\B(\H;\H)$ and $\K(\H)$ for $\K(\H;\H)$.

For complex numbers, we use the standard notation $\C_\pm :=\{z \in \C \mid \pm\Im z>0\}$.

\vspace{10pt}

{\sc Acknowledgments}

The authors would like to thank the anonymous referee for
careful reading and
constructive suggestions which enable them to improve the paper.


\section{$\boldsymbol{0}$-energy threshold}\label{sec0behavior}
\setcounter{equation}{0}

In this section, we derive various results about the behavior of the resolvent of $H_0$ at $0$. We also provide information about the $0$-energy eigenvalue of $H$ and about the absence of $0$-energy resonance for this operator.
Finally, we show that if $0$ is not an eigenvalue of $H$,
then this operator can not have an accumulation of positive eigenvalues at $0$.

We start by studying an auxiliary operator which will be related to the behavior of the resolvent of $H_0$ at $0$.
Following \cite[Sec.~2]{U}, let us set $G_0$ for the operator defined for $f\in C_c^\infty(\R^3)$ by
\begin{equation*}
[G_0f](x) :=\frac{1}{2\pi^2}\int_{\R^3}\frac{1}{|x-y|^2}f(y)\;\!\d y.
\end{equation*}
Clearly, this corresponds to the operator of convolution by the function
\begin{equation}\label{defg0}
g_0: \R^3 \to \R \quad \hbox{with}\quad  g_0(x):= \frac{1}{2\pi^2}|x|^{-2}.
\end{equation}
It has been shown in \cite[Lem.~5.1]{U} that this operator continuously extends
to an element of $\B(\H_s,\H)$ as well as an element of $\B(\H,\H_{-s})$ for any $s>3/2$.
The following statement is an improvement of this result.

\begin{lemma}\label{lemG0}
For any $s>1$, the operator $G_0$ belongs to $\K(\H_s,\H)$ and to $\K(\H,\H_{-s})$.
\end{lemma}

\begin{proof}
Let us set $\varrho$ for $\langle \cdot\rangle^{-s}$ with $s>1$.
Clearly, one has to show that the operators $G_0\varrho(X)$ and $\varrho(X)G_0$ belong to $\K(\H)$. However, since $\varrho \in L^3(\R^3)$ and since the map $\R^3\ni x \mapsto|x|^{-1}\in \R$ belong to the weak $L^p$-space $L^3_w(\R^3)$, these inclusions follow directly from \cite[Thm.~4.2]{Simon}.
\end{proof}

It clearly follows from this result that $G_0$ belongs to $\K(\H_s,\H_{-s})$ for any $s>1$. In fact, by real interpolation  one also obtains that the operator $G_0$ belongs to $\K(\H_{(1-\theta) s},\H_{-\theta s})$ for any $\theta \in [0,1]$.
Indeed, this result follows from \cite{H} together with the identification of the interpolation spaces
$S(\theta, 2; \H_s, \H)$, resp.~$S(\theta, 2; \H, \H_{-s})$, introduced in that reference with $\H_{(1-\theta)s}$, resp.~$\H_{-\theta s}$ (see also \cite[Sec.~2.8.1]{ABG} for additional information on real interpolation).
In particular, by choosing $\theta =1/2$, one deduces that $G_0$ belongs to $\K(\H_{s},\H_{-s})$, for any $s>1/2$.

Now, it is shown in  \cite{BN} that the resolvents $(H_0-\lambda \mp i\varepsilon)^{-1}$ admit limits as $\varepsilon \searrow 0$ in $\B(\H_{s},\H_{-s})$ for any $s>1/2$ and $\lambda>0$. In that respect, it would be interesting to understand the behaviour of $(H_0-\lambda \mp i0)^{-1}\in \B(\H_s,\H_{-s})$ for $\lambda \searrow 0$ and for $s>1/2$. In the next lemma, we obtain such a description but our approach is valid only in $\K(\H_{s},\H_{-s})$  with $s>1$. It is not clear to us if the convergences still hold in $\K(\H_{s},\H_{-s})$ with $s>1/2$.

\begin{lemma}\label{funclimit}
For any $s>1$ and $\lambda \in (0,\infty)$, the operators $(H_0-\lambda \mp i0)^{-1}$ belong to $\K(\H_s,\H_{-s})$. Furthermore, the maps $(0,\infty)\ni \lambda \mapsto  (H_0-\lambda \mp i0)^{-1} \in \K(\H_s,\H_{-s})$ are continuous in norm and converge to $G_0$ as $\lambda \searrow 0$.
\end{lemma}

\begin{proof}
Recall from \cite[Eq.~(5.3)]{U} that for any $\lambda>0$ the following formal equalities hold:
\begin{equation}\label{sum3}
R_0(\lambda \pm i0):=(H_0-\lambda\mp i 0)^{-1} = G_0 + K_\lambda^\pm +M_\lambda,
\end{equation}
where the definitions of $K_\lambda^\pm$ and of $M_\lambda$ are going to be recalled below. Thus, the present proof consists first in introducing the rigorous meaning of \eqref{sum3} and then in showing that for $s>1$ the operators $K_\lambda^\pm$ and  $M_\lambda$ belong to $\K(\H_s,\H_{-s})$, that they are continuous in norm as functions of $\lambda$, and that they converge in norm to $0$ as $\lambda \searrow 0$. Equivalently, one can show the same properties for the operators $\langle X \rangle^{-s}K_\lambda^\pm\langle X \rangle^{-s}$ and $\langle X \rangle^{-s}M_\lambda\langle X \rangle^{-s}$ in $\K(\H)$.

It has been proved in \cite[Eq.~(4.14)]{U} that $R_0(\lambda\pm i0)f= G_\lambda^\pm f$ for any $f \in C^\infty_c(\R^3)$, where $G_\lambda^\pm$ are the integral operators defined by
\begin{equation*}
[G_\lambda^\pm f](x):= \int_{\R^3} g_\lambda^\pm (x-y)\;\!f(y)\;\!\d y
\end{equation*}
with
\begin{equation}\label{defg}
g_\lambda^\pm(x):=\frac{1}{2\pi^2|x|^2} + k_\lambda^\pm(x) + m_\lambda(x)
\end{equation}
and
\begin{align}
&k_\lambda^\pm(x):=\frac{\lambda}{2\pi}\cdot \frac{e^{\pm i \lambda |x|}}{|x|}, \nonumber\\
&m_\lambda(x):= \frac{\lambda}{2\pi^2|x|}\big(
\sin(\lambda|x|)\;\ci(\lambda|x|)+ \cos(\lambda|x|)\;\si(\lambda|x|)\big),  \label{defg1}
\end{align}
where $\ci$ and $\si$ are respectively the cosine integral and the sine integral functions.
Note that these expressions explicitly define each term in \eqref{sum3}.

Now, let us observe that
$K_\lambda^\pm =2\lambda(-\Delta-\lambda^2 \mp i 0)^{-1}$.
It is well known (see for example \cite{JK}) that the map
$z \mapsto (-\Delta-z)^{-1}\in \B(\H^{-1}_{s},\H^1_{-s'})$
is continuous for $z \in \overline{\C_\pm}$ and for any
$s,s'>1/2$ with $s+s'>2$.
In particular, this resolvent is continuous as $z \to 0$ in $\overline{\C_\pm}$. Then,
by an adequate choice of $s$
and $s'$, one infers that the maps
$\R \ni \lambda \mapsto \langle X \rangle^{-s} K_\lambda^\pm \langle X \rangle^{-s}\in \K(\H)$
are continuous in norm and that
$\lim_{\lambda\searrow 0} \langle X \rangle^{-s} K_\lambda^\pm \langle X \rangle^{-s}=0$ in norm.

For the compactness of the operator $\langle X \rangle^{-s} M_\lambda \langle X \rangle^{-s}$ for $\lambda>0$, let us set $\varrho(\cdot)$ for $\langle \cdot \rangle^{-s}$ for some $s>3/2$. By taking the estimate \cite[Eq.~(5.16)]{U} into account, namely
\begin{equation}  \label{defg2}
\big|
\sin(r)\;\ci(r)+ \cos(r)\;\si(r)\big|\leq {\rm Const.}~(1+r)^{-1},  \quad  \ \
  0<  r <\infty,
\end{equation}
it is easily seen that the function $m_\lambda$ belongs to $L^2(\R^3)$ and thus the operator $\varrho(X)M_\lambda$ is a Hilbert-Schmidt operator.
Then, let us observe that the relation $m_\lambda(x)= \lambda^2 m_1(\lambda x)$ holds for any $\lambda>0$. One deduces that
\begin{align}
\nonumber \|\varrho(X) M_\lambda - \varrho(X) M_{\lambda'}\|_{\B(\H)}
&\leq \|\varrho(X) M_\lambda - \varrho(X) M_{\lambda'}\|_{HS} \\
\nonumber &= \|\varrho\|_{L^2(\R^3)}\;\|m_\lambda-m_{\lambda'}\|_{L^2(\R^3)} \\
\label{diff}&=\|\varrho\|_{L^2(\R^3)}\;\big\|\lambda^2 m_1(\lambda\cdot)-(\lambda')^2 m_1\big(\lambda {\textstyle \frac{\lambda'}{\lambda}}\cdot \big)\big\|_{L^2(\R^3)}
\end{align}
and that \eqref{diff} vanishes as $\lambda'\to \lambda$ because of the continuity of the dilation group in $L^2(\R^3)$.
Finally, from the equality $\|m_1(\lambda\cdot)\|_{L^2(\R^3)} = \lambda^{-3/2}\|m_1\|_{L^2(\R^3)}$ one infers that
\begin{align*}
\|\varrho(X) M_\lambda\|_{HS} & = \|\varrho\|_{L^2(\R^3)}\;\|m_\lambda\|_{L^2(\R^3)} \\
&= \lambda^2 \|\varrho\|_{L^2(\R^3)}\;\|m_1(\lambda\cdot)\|_{L^2(\R^3)}\\
&=\lambda^{1/2} \|\varrho\|_{L^2(\R^3)}\;\|m_1\|_{L^2(\R^3)}
\end{align*}
which implies that $\|\varrho(X) M_\lambda\|_{\B(\H)}\leq {\rm Const.}~\lambda^{1/2}$.

Clearly, the same estimates and results hold for the operator $M_\lambda \varrho(X)$. Thus, one has obtained that $M_\lambda\in \K(\H_{s},\H) \cap \K(\H,\H_{-s})$ for any $s>3/2$, and that the norm of this operator is continuous in $\lambda$ and vanishes as $\lambda^{1/2}$ when $\lambda\searrow 0$ in both norms. By a real interpolation argument, one obtains that the same result holds in $\K(\H_1,\H_{-1})$. Note that the control on the dependence on $\lambda$ for the norm in $\K(\H_1,\H_{-1})$ can be obtained by taking \cite[Eq.~(2.6.2)]{ABG} into account.
\end{proof}

In Propositions \ref{prp:sparse1} and \ref{prp:sparse2} below,
we show that $0$ is
generically not an eigenvalue of the operator
$H$. To this end, we follow the arguments presented in
\cite{BE0} in the context of Weyl-Dirac operators.
For that purpose, we introduce the set
$L^3(\R^3;\R)$
as a natural class for the potential $V$. Note that any measurable and real function $V$ satisfying Condition \eqref{condition1} with $\sigma>1$ belongs to $L^3(\R^3;\R)$.

\begin{lemma} \label{lem:sa}
If $V\in L^3({\R}^3;\R)$, then $V$ is $H_0$-bounded
with relative bound $0$. In particular, $H:= H_0 +V$ is a self-adjoint
operator in $\H$ with domain $\H^1$.
\end{lemma}

The proof of Lemma \ref{lem:sa} can be mimicked from the proof of \cite[Lem.~2]{BE0}.

\begin{lemma} \label{lem:cpt}
If $\;V\in L^3({\R}^3; \R)$,
then $(-\Delta)^{-1/4}V(-\Delta)^{-1/4}$ is a compact
operator in $\H$ satisfying
\begin{equation}  \label{eq:cpt1}
\big\|(-\Delta)^{-1/4}V(-\Delta)^{-1/4} \big\|_{\B(\H)}
\leq 2^{-1/3}{\pi}^{-2/3}\;\!\|V \|_{L^3(\R^3)}.
\end{equation}
\end{lemma}

\begin{proof}
We only prove the inequality \eqref{eq:cpt1}, the proof of the compactness can be mimicked directly from the proof of \cite[Lem.~1]{BE0}.

To prove the inequality \eqref{eq:cpt1}, we first borrow
the Sobolev inequality for $\sqrt{-\Delta}$
from \cite[Sec.~8.4]{LL} or from \cite[p.119, Thm.~1]{Stein}, namely
that for any $f \in \H$:
\begin{equation}  \label{eq:cpt2}
\| f \|_{L^2(\R^3)}
\ge 2^{1/6} {\pi}^{1/3}
\big\|(-\Delta)^{-1/4} f \big\|_{L^3(\R^3)}.
\end{equation}
Now let $f$, $g\in \SS({\R}^3)$.
We then see that $(-\Delta)^{-1/4} f\in {\H}^1$, hence
$V(-\Delta)^{-1/4} f\in {\H}$ by Lemma \ref{lem:sa}
and
$(-\Delta)^{-1/4} V (-\Delta)^{-1/4} f \in L^3(\R^3)$
by \eqref{eq:cpt2}.
Therefore, one can appeal to the definition of the
Fourier transform of tempered distributions, and gets
\begin{align}
\nonumber \big\langle g, \, (-\Delta)^{-1/4} V (-\Delta)^{-1/4} f
\big \rangle_{\SS, {\SS}'}
&=
\big \langle \FF g, \, \FF[(-\Delta)^{-1/4} V (-\Delta)^{-1/4} f] \big \rangle_{\SS, {\SS'}} \\
\nonumber &=
\int_{{\R}^3} |k|^{-1/2} \,(\FF g)(k)
 \, \overline{\FF \big( V
 (-\Delta)^{-1/4}f\big)(k)} \, \d k  \\
&=
\int_{{\R}^3} \big((-\Delta)^{-1/4}g\big)(x)
 \, \overline{V(x)
\big( (-\Delta)^{-1/4}f\big)(x)} \, \d x.    \label{eq:cpt3}
\end{align}
By applying H\"older inequality twice to \eqref{eq:cpt3},
one obtains
\begin{gather}  \label{eq:cpt5}
\begin{split}
\big|\big\langle g, \, (-\Delta)^{-1/4} V (-\Delta)^{-1/4} f
\big\rangle_{\SS, {\SS}'} \big|
&\le
\big\| (-\Delta)^{-1/4} g \big\|_{L^3(\R^3)}
\;\!\| V \|_{L^3(\R^3)}
\;\! \big\| (-\Delta)^{-1/4}f \big\|_{L^3(\R^3)}  \\
&\le 2^{-1/3} {\pi}^{-2/3}  \|g \|_{L^2(\R^3)}
\;\!\|V \|_{L^3(\R^3)} \;\!\|f \|_{L^2(\R^3)}.
\end{split}
\end{gather}
In the second inequality of \eqref{eq:cpt5}, we have used \eqref{eq:cpt2}.
Since $\SS({\R}^3)$ is dense
in $\H$, we find that
the inequality \eqref{eq:cpt5} is valid for all $g\in \H$. Hence, it follows that $(-\Delta)^{-1/4} V (-\Delta)^{-1/4} f \in \H$, and the  estimate \eqref{eq:cpt1} is then obtained by  density argument.
\end{proof}

For the next statements, we need the notation
$H_V:=H_0 + V$ to indicate the dependence on $V$.
Let us also denote
by $\sigma_p(H_V)$ the point spectrum of $H_V$,
by $\Ker(H_V)$ the subset
$\{f\in \H^1 \mid H_V f = 0 \}$ and by $\Nul(H_V)$  the dimension of this subset.

\begin{proposition} \label{prp:sparse1}
Let $V$ be in $L^3({\R}^3;\R)$.
Then
$0 \not \in \sigma_p(H_0+aV)$
for all $a\in \R$ except for a discrete subset of $\R$.
\end{proposition}

\begin{proof}
Let us define a $\B(\H)$-valued analytic function on $\C$
by
\begin{equation*}
K_{z}:=(-\Delta)^{-1/4}zV(-\Delta)^{-1/4}=z(-\Delta)^{-1/4}V(-\Delta)^{-1/4}.
\end{equation*}
By Lemma \ref{lem:cpt}, $K_{z}$ is a compact operator for each $z \in \C$. Therefore, one can apply the analytic Fredholm theorem (see for example \cite[p.~201]{RS1})
and deduce that
$(I + K_{z})$ is invertible in $\B(\H)$ for all $z \in \C$
except for a discrete subset of $\C$.
In particular, one infers that $(I + K_{a})^{-1}\in \B(\H)$ for all $a \in \R$ except for a discrete subset of $\R$.

Now, let $a\in \R$ such that
$(I + K_{a})^{-1}\in \B(\H)$,
and let us assume that there exists $f \in \Ker(H_{aV})$. Clearly, one has
$H_0f= -aVf$. Then, let us set $g:= (-\Delta)^{1/4}f  \in \H$ which satisfies
\begin{equation}  \label{eq:sprs2}
g = - (-\Delta)^{-1/4}aV(-\Delta)^{-1/4} g.
\end{equation}
It is obvious that \eqref{eq:sprs2} is
equivalent to $(I + K_{a})g=0$.
This implies that $g =0$, because $(I + K_{a})$ is invertible in $\B(\H)$. Since $(-\Delta)^{1/4}$
is an injective mapping from $\H^1$ to $\H$, it follows that $f=0$,  and we can conclude that $\Nul(H_{aV}) = 0$
whenever $(I + K_{a})$ is invertible in $\B(\H)$.
\end{proof}

\begin{proposition} \label{prp:sparse2}
The set
$\V:=\{ V \in L^3({\R}^3; \R) \mid 0 \not \in \sigma_p(H_0+V) \}$
contains an open and dense subset of $L^3({\R}^3;\R)$.
\end{proposition}

\begin{proof}
Let us now set
\begin{equation*}
K_V:=(-\Delta)^{-1/4}V(-\Delta)^{-1/4}
\end{equation*}
for any $V \in L^3({\R}^3; \R)$.
In the sequel, we show that the set $\widetilde{\V}$ defined by
\begin{equation*}
\widetilde{\V}:=
\big \{ V \in L^3({\R}^3; \R) \mid
 \Ker(I + K_V) =\{0\} \big \}
\end{equation*}
is open and dense in $L^3({\R}^3; \R)$.
Then, the statement of the Proposition is a consequence of the inclusion $\widetilde \V \subset \V$ which has already been proved in the second half of the previous proof.

Let $V \in \widetilde{\V}$. Since $K_V$ is a compact operator
in $\H$ by Lemma \ref{lem:cpt}, we observe that
$(I + K_V)$ is invertible in $\B(\H)$.
Now choose
a real number  $\delta > 0$ small enough such that
$\|(I + K_V)^{-1}\| \, 2^{-1/3}{\pi}^{-2/3}\delta <1$.
If $V' \in L^3({\R}^3; \R)$
satisfies $\| V - V^{\prime} \|_{L^3(\R^3)} < \delta$,
then the identity
\begin{equation*}
I + K_{V'}
=
(I + K_{V}) \big(
I + (I + K_V)^{-1}
(K_{V'} - K_{V})
\big),
\end{equation*}
together with \eqref{eq:cpt1},
enables one to construct the inverse of $I + K_{V'}$
by a Neumann series.  Hence,
$V' \in \widetilde{\V}$, and then $\widetilde{\V}$ is an open subset of $L^3({\R}^3; \R)$.

To prove the density of $\widetilde{\V}$,
let $\varepsilon>0$ and $V\in L^3({\R}^3; \R)$ be given.
It then follows from the proof of Proposition \ref{prp:sparse1}
that $(I + K_{aV})$ is invertible in $\B(\H)$ for all $a \in \R$ except for
a discrete subset of $\R$.
This means that one can choose $a \in \R$
so that $\|V - aV \|_{L^3(\R^3)} < \varepsilon$ and
that $\Ker(I + K_{aV}) =\{0\}$.
Therefore $aV \in \widetilde \V$, and then $\widetilde{\V}$ is dense in $L^3({\R}^3; \R)$.
\end{proof}

We conclude this section with a theorem, which asserts that if $0$ is not an eigenvalue of $H$, then $0$ cannot be an accumulation
point of positive eigenvalues of $H$. On the way, we also discuss about the non-existence of $0$-energy resonances.

\begin{theorem}  \label{prp:0eigenvl}
Assume that $V$ satisfies Condition \eqref{condition1} with $\sigma > 2$, and that $0 \not \in \sigma_p(H)$. Then there exists a constant $\lambda_0>0$
such that $[0, \, \lambda_0) \cap \sigma_p(H)= \emptyset$.
\end{theorem}

To prove this assertion, we need a preliminary lemma. For its statement, we use a convenient decomposition of the potential $V$ into two parts, namely $V=uv$ with
\begin{equation}\label{decompV}
v := |V|^{1/2} \qquad {\rm and} \qquad
u := |V|^{1/2} \;{\rm sgn}(V).
\end{equation}

\begin{lemma}  \label{lem:0eigenvl-1}
Assume that $V$ satisfies Condition \eqref{condition1} with $\sigma > 2$, and that $0 \not \in \sigma_p(H)$. Then, the operator $I + uG_0v$ is invertible in $\B(\H)$.
\end{lemma}

\begin{proof}
From Lemma \ref{lemG0}, we know that $uG_0v$ belongs to $\K(\H)$.
Therefore, it is sufficient to show that $-1$ is not
an eigenvalue of $uG_0v$.

By contradiction, let us assume that there exists $f \in \H$ satisfying  $uG_0v f = -f$, and let us set $g:=G_0vf\in \H$.
One infers then that
$$
G_0Vg = G_0vu(G_0vf) = G_0v(-f)=-g.
$$
Thus, for any $h\in \H^1$ one has
\begin{equation*}
-\langle g, H_0h\rangle_\H= \langle G_0Vg,H_0h\rangle_\H
= \langle \FF G_0Vg,\FF H_0h\rangle_\H
=\big\langle |X|^{-1}\FF Vg, |X| \FF h\big\rangle_\H
= \langle Vg,h\rangle_\H,
\end{equation*}
where we have used the equality $(-\Delta)^{-1/2}=G_0$ as operators on $\H$, see \cite[Corol.~4.5]{MSP99}.
Thus, one deduces that $\langle g, H_0 h\rangle_\H= \langle - Vg,h\rangle_\H$ for any $h\in \H^1$. By the definition of the adjoint of an operator and since $H_0$ is self-adjoint, one directly infers from this equality that $g$ belongs to the domain of $H_0$ and satisfies $H_0g=-Vg$, or equivalently that $g\in \H^1$ and $Hg=0$. By the assumption, it follows that $g=0$, from which one deduces that $f=-ug=0$.
\end{proof}

We should like to stress that Lemma \ref{lem:0eigenvl-1} points towards the non-existence of $0$-energy resonances for $H$. Indeed, $0$-energy bound states or $0$-energy resonances usually prevent the expression $I+uG_0v$ of being invertible: In a vague and formal sense, the function $g$ of the proof of Lemma \ref{lem:0eigenvl-1} is a solution of the equation $Hg=0$, and these solutions usually prevent the invertibility of the mentioned operator. However, Lemma  \ref{lem:0eigenvl-1} asserts that $g=0$ if
$0 \not \in \sigma_p(H)$, and therefore $0$-energy bound states are the only troublesome solutions of the formal equality $Hg=0$.

\begin{proof}[Proof of Theorem \ref{prp:0eigenvl}]
Let us first recall that map
$$
[0,\infty)\ni \lambda\mapsto uR_0(\lambda\pm i0)v\in\B(\H)
$$
is continuous in norm and converges to $uG_0v$ as $\lambda\searrow 0$, see Lemma \ref{funclimit}. In addition, since the set of invertible elements in $\B(\H)$ is open, one infers from Lemma \ref{lem:0eigenvl-1} that there exists a positive constant $\lambda_0$ such that
for each $\lambda \in [0, \, \lambda_0)$, the operators $I + u R_0(\lambda \pm i 0)v$ are invertible in $\B(\H)$. In particular, it implies that $-1$ does not belong to the spectrum of $uR_0(\lambda\pm i0)v$ for any $\lambda\in [0,\lambda_0)$.

To prove the statement, it remains to show that the previous sentence implies the absence of eigenvalue of $H$ in $[0,\lambda_0)$. In fact, such a relation between the eigenvalues of $H$ and the eigenvalue $-1$ for $uR_0(\lambda\pm i0)v$ is well known, see for example \cite[Lem.~4.7.8]{Yaf}. Note that the assumptions of this lemma are fulfilled, since the strong $H_0$-smoothness of $v$ and $u$ (\cf~\cite[Def.~4.4.5]{Yaf} for this definition) is a standard consequence of the so-called trace theorems \cite[Thm.~1.1.4 \& 1.1.5]{Yaf2}.
\end{proof}


\section{Stationary expression for the wave operators}\label{secsta}
\setcounter{equation}{0}

In this section we derive stationary expressions for the wave operators which were already announced in \cite{U0}. Since the limiting absorption principle and
the generalized eigenfunction expansions for the operator $H$ was
established in \cite{U}, we can follow the line of \cite[Sec.~2]{KK}, where the discussions were made in an abstract setting,
and the line of \cite[Chapt.~5]{K}, where the discussions
were given for the three-dimensional Schr\"odinger
operator.

For $z \in \C\setminus\R$, let us recall that $R_0(z)$ and $R(z)$ are used respectively for the resolvents $(H_0-z)^{-1}$ and $(H-z)^{-1}$. The notation $E_0(\cdot)$ is used for the spectral measure of $H_0$. We also recall that the following limiting absorption principle has been proved in \cite{BN}, namely for $s>1/2$ and $\lambda\in (0,\infty)\setminus \sigma_p(H)$ the operators $R(\lambda\pm i 0):=\lim_{\varepsilon \searrow 0}R(\lambda \pm i \varepsilon)$ belong to $\B\big(\H_s,\H_{-s}\big)$. Note that the condition $\sigma>1$ in \eqref{condition1} has been tacitly assumed. As a consequence, the wave operators $W_\pm$ defined by the strong limits
\begin{equation*}
W_\pm:=s-\lim_{t \to \pm \infty} e^{itH}e^{-itH_0}
\end{equation*}
exist and are asymptotically complete. In addition, these expressions are equal to the ones obtained by the usual stationary approach, see for example \cite[Thm.~5.3.6]{Yaf}.

\begin{lemma}\label{lem:wosr1}
Let $s \in (1/2,\sigma-1/2)$ and assume that $f$, $g$ belong to $\H_s$ with $E_0([a,b])g=g$ for some $[a,b] \subset (0,\infty) \setminus \sigma_p(H)$.
Then one has
\begin{equation}\label{waveop+u3}
\langle f, W_\pm g\rangle _\H =
\int_a^b \big \langle \{1 - V R(\lambda\pm i0) \}f ,
E_0^{\prime}(\lambda)g \big \rangle_{s, -s} \,\d\lambda\ ,
\end{equation}
where $E_0^{\prime}(\lambda) := \frac{1}{2\pi i}
\big( R_0(\lambda+i0) - R_0(\lambda-i0) \big)\in \B\big(\H_s,\H_{-s}\big)$.
\end{lemma}

Note that it follows from the hypothesis on $g$ that
$E_0(\Delta)g=0$ for any Borel set
$\Delta \subset J:=\R \setminus [a, b]$,
and that $E_0^{\prime}(\lambda)g =0$ for all $\lambda \in J$.
Thus, the usual integral over $\R$ reduces to an integral over the finite interval $[a,b]$. The following proof is standard, but we recall it for completeness.

\begin{proof}
Let $\varepsilon >0$ and $f,g$ as in the statement. By Parseval's identity and the equation of the resolvent
$R(z) = R_0(z) \{ 1 - VR(z)\}$, one has:
\begin{align} \label{eqini}
\int_0^{\infty}
\big \langle e^{-\varepsilon t} e^{\mp iHt}f,e^{-\varepsilon t} e^{\mp iH_0t}g\big\rangle_{\H}  \,\d t
&=
\frac{1}{2\pi} \int_{-\infty}^{\infty}
\big\langle
R(\lambda \pm i \varepsilon)f, R_0(\lambda \pm i \varepsilon)g
\big\rangle_{\H} \,\d\lambda \\
\label{waveop+u5}
&=
\frac{1}{2\varepsilon} \int_{-\infty}^{\infty}
\big\langle
\{1 - V R(\lambda \pm i\varepsilon) \} f, \delta_{\varepsilon}(H_0-\lambda)g
\big\rangle_{\H} \, \d\lambda,
\end{align}
where
\begin{equation*}
\delta_{\varepsilon}(H_0-\lambda) =
 \frac{\varepsilon}{\pi} R_0(\lambda \pm i \varepsilon)^*R_0(\lambda\pm i \varepsilon)
=\frac{1}{2\pi i} \big(R_0(\lambda + i \varepsilon)-
R_0(\lambda - i \varepsilon)\big).
\end{equation*}
Furthermore, it is known that since the wave operators exist they are also obtained by the Abelian limit
\begin{equation}\label{waveop+u11}
\big\langle f, W_\pm g\big\rangle_\H=
\lim_{\varepsilon \searrow 0}
2\varepsilon
\int_0^{\infty}
\big\langle e^{-\varepsilon t} e^{\mp iHt}f,e^{-\varepsilon t} e^{\mp iH_0t}g\big\rangle_{\H}  \,\d t.
\end{equation}
Thus, by combining \eqref{waveop+u5} and \eqref{waveop+u11}, one gets
\begin{equation}\label{waveop+u13}
\big\langle f,W_\pm g\big\rangle_{\H}=
\lim_{\varepsilon \searrow 0}
\int_{-\infty}^{\infty}
\big\langle
\{1 - VR(\lambda \pm i\varepsilon)\}f, \delta_{\varepsilon}(H_0-\lambda)g
\big\rangle_{\H} \,\d \lambda.
\end{equation}

Now, it follows from the limiting absorption principle recalled before the statement of the lemma that for each $\lambda \in (0,\infty) \setminus \sigma_p(H)$ one has
\begin{equation}\label{eqlimit}
\lim_{\varepsilon \searrow 0}
\big\langle
\{1 - VR(\lambda \pm i\varepsilon)\}f, \delta_{\varepsilon}(H_0-\lambda)g
\big\rangle_{\H}
=
 \big \langle \{1 - V R(\lambda\pm i0) \}f ,
E_0^{\prime}(\lambda)g \big \rangle_{s, -s}\ .
\end{equation}
Thus, the statement of the lemma is obtained once the permutation of the integral and the limit in \eqref{waveop+u13} is justified.
For that purpose, recall that
\begin{equation}  \label{eqlimit+}
\frac{\varepsilon}{\pi} \big\langle
R(\lambda \pm i \varepsilon)f, R_0(\lambda \pm i \varepsilon)g\big\rangle_{\H}
=
\big\langle
\{1 - VR(\lambda \pm i\varepsilon)\}f, \delta_{\varepsilon}(H_0-\lambda)g
\big\rangle_{\H}.
\end{equation}
Then \eqref{eqlimit} and \eqref{eqlimit+}
enable us to
apply \cite[Lem.~5.2.2]{Yaf} which justifies the permutation and thus leads directly to the statement of the lemma. Note that the change of the two endpoints in the integral is also a consequence of that abstract result.
\end{proof}

In the next lemma we derive an explicit expression for the operator $E_0^\prime(\lambda)$ for any $\lambda \in (0,\infty)$. Before this,  let us simply recall that if $f \in \H_s$ with $s>3/2$ then $f$ belongs to $L^1(\R^3)$ and thus its Fourier transform $\hat f$ belongs to $C_0(\R^3)$.

\begin{lemma}\label{lem:trace}
For any $f,g\in \H_s$ with $s > 3/2$ and for any $\lambda>0$ one has
\begin{equation}\label{eqn:trace}
\big\langle f, E_0^{\prime}(\lambda) g \big\rangle_{s, -s}
= \lambda^2 \big\langle
\gamma(\lambda)  \hat f, \gamma(\lambda) \hat g\big\rangle_{L^2({\S}^2)}
\end{equation}
with $\gamma (\lambda)$ the trace operator onto the sphere $\big\{k\in \R^3 \mid |k|= \lambda \big \}$, {\it i.e.}~$
\big(\gamma (\lambda) \hat f\big) (\omega):= \hat f (\lambda \omega)$ for any $\omega \in \S^2$.
\end{lemma}

\begin{proof}
For this proof, we use the integral kernels $g_{\lambda}^{\pm}(x-y) $ of the extended resolvents $R_0(\lambda\pm i0)$ obtained in \cite[Sect.~4]{U} and already recalled in \eqref{defg}.
Thus, for $f,g$ as in the statement and $x \in \R^3$ one has
\begin{eqnarray}\label{eqn:3q1}
\nonumber \big[\big(R_0(\lambda+i0) - R_0(\lambda-i0) \big)g\big](x)
&=&\frac{\lambda}{2\pi}\int_{\R^3}
\frac{e^{i \lambda |x-y|} -e^{-i \lambda |x-y|}}{|x-y|}
g(y)  \,\d y   \\
&=& \frac{\lambda i}{\pi} \int_{\R^3} \frac{\sin (\lambda |x-y|)}{|x-y|} g(y) \,\d y\ .
\end{eqnarray}
It then follows from the definition of $E_0^\prime(\lambda)$ and from \eqref{eqn:3q1} that
\begin{eqnarray}\label{train1}
\nonumber \big\langle  f, E_0^{\prime}(\lambda) g \big\rangle_{s, -s}
&=& \frac{1}{2\pi i}
\big \langle f, \big( R_0(\lambda+i0) - R_0(\lambda-i0) \big) g \big\rangle_{s, -s}   \\
&=& \frac{\lambda}{2\pi^2}
\iint_{\R^6}
f(x) \frac{\sin (\lambda|x-y|)}{|x-y|}
\overline{g(y)}  \,\d x\, \d y .
\end{eqnarray}
By appealing to the formula
\begin{equation*}
\int_{\S^2} e^{-i \lambda \omega \cdot x}  \, \d \omega
= \frac{4\pi \sin (\lambda |x|)}{\lambda |x|}
\end{equation*}
and by the change of order of integration (valid because of our assumptions on $f$ and $g$) we get
\begin{eqnarray*}
\iint_{\R^6} f(x) \frac{\sin (\lambda|x-y|)}{|x-y|}
\overline{g(y)} \, \d x\, \d y
&=& \frac{\lambda}{4\pi} \int_{\S^2}\big\{
\int_{\R^3}e^{-i \lambda \omega \cdot x}  f(x) \, \d x \big\}
\overline {\big\{
\int_{\R^3}e^{-i \lambda \omega \cdot y}  g(y) \, \d y
\big\}}  \, \d \omega   \\
&=&
2\pi^2 \lambda \int_{\S^2} \hat f (\lambda \omega)
\,\overline{\hat g (\lambda \omega)} \, \d\omega  \\
&=& 2\pi^2 \lambda
\big\langle\gamma(\lambda)  \hat f, \gamma(\lambda) \hat g)\big\rangle_{L^2(\S^2)}.
\end{eqnarray*}
By combining these equalities with \eqref{train1} one directly obtains \eqref{eqn:trace}.
\end{proof}

Let us now define the generalized Fourier transforms by the relations
\begin{equation}  \label{eqn:gft1}
\FF_{\pm}:= \FF W_{\pm}^*\ .
\end{equation}
In the next lemma, we derive standard and more explicit formulas for $\FF_\pm$.

\begin{lemma}  \label{lem:gef1}
Assume that $\sigma>2$ and let $s \in (3/2,\sigma-1/2)$. Let  $f,g \in \H_s$ with $E_0([a,b])g=g$ for some $[a,b] \subset (0,\infty) \setminus \sigma_p(H)$. Then
\begin{equation}  \label{eqn:gft2}
\big\langle \FF_{\pm} f, \, \FF g \big\rangle_\H
=
\int_{\R^3}
\big[\FF \{ 1 - V R(|k|\pm i0) \} f \big](k)
\, \overline{\FF[g](k)} \, \d k.
\end{equation}
\end{lemma}

\begin{proof}
It follows successively from \eqref{eqn:gft1}, \eqref{waveop+u3}  and \eqref{eqn:trace} that
\begin{eqnarray*}\label{eqn:gft2-2}
\big\langle \FF_{\pm} f, \FF g \big\rangle_\H
&=&
\int_a^b \big\langle  \{ 1 - V R(\lambda \pm i0) \}f ,
E_0^{\prime}(\lambda)g \big \rangle_{s, -s} \, \d\lambda  \\
&=& \int_a^b  \lambda^2
\big\{ \int_{\S^2}
\big[\FF\{1 - V R(\lambda\pm i0)\}f\big](\lambda \omega)
\,\overline{\FF[g](\lambda \omega)} \, \d \omega
\big\} \, \d\lambda .
\end{eqnarray*}
By the change of the variables $k:= \lambda \omega$, one obtains the result.
\end{proof}

For fixed $k \in \R^3$ let us now define $\varphi_0(\cdot,k)$ by $\varphi_0(x,k) := e^{i k \cdot x}$. Clearly, $\varphi_0(\cdot,k)\in \H_{-s}$ for any $s>3/2$.
Since the subset of all elements $g$ satisfying the condition of the previous lemma is dense in $\H$, it follows from \eqref{eqn:gft2} that for $f \in \H_s$ and almost every $k$ one has
\begin{eqnarray*}
[\FF_{\pm} f](k)
&=& \big[ \FF\big( 1 - V R(|k|\pm i0) \big) f \big](k) \\
&=& (2\pi)^{-3/2} \big\langle
\{ 1 - V R(|k|\pm i0) \} f,  \varphi_0(\cdot, k)
\big\rangle_{s, -s}  \\
&=& (2\pi)^{-3/2} \big \langle   f,  \{ 1 - R(|k|\mp i0)V \} \varphi_0(\cdot, k)
\big\rangle_{s, -s} \\
&=& (2\pi)^{-3/2} \int_{\R^3} \overline{\varphi^{\pm}(x, k)} f(x)  \, \d x,
\end{eqnarray*}
where we have used the definition of the generalized eigenfunctions $\varphi^\pm(\cdot,k)$ introduced in \cite[Eq.~(8.5)]{U}:
$$
\varphi^\pm(\cdot,k):=  \{ 1 - R(|k|\mp i0)V \} \varphi_0(\cdot, k)\ .
$$

Now, it follows from \eqref{eqn:gft1} that
$\FF_{\pm}^* = W_{\pm} \FF^*$, or equivalently $W_{\pm} = \FF_\pm^* \FF$.
Thus for $f$, $g$ as in the previous lemma one infers that
\begin{eqnarray*}
\big\langle f, W_\pm g\big \rangle_\H
&=&
\big\langle \FF_{\pm} f, \FF g\big\rangle_\H  \\
&=&
\int_{\R^3} \big\{(2\pi)^{-3/2} \int_{\R^3}  \overline{\varphi^{\pm}(x,k)} f(x)  \, \d x\big\}\, \overline{\hat g(k)} \, \d k  \\
&=&
\int_{\R^3} \big\{\overline{ (2\pi)^{-3/2} \int_{\R^3} {\varphi^{\pm}(x, k)}  \hat g(k) \, \d k}\big\}\,
f(x) \, \d x   .
\end{eqnarray*}
Note that for the interchange of the integrals, one has used that $\varphi^\pm$ satisfy the following bound \cite[Thm.~9.1]{U}: for any compact set $K\subset (0,\infty)\setminus \sigma_p(H)$ there exists $c= c(K)$ such that
$$
\sup_{x\in {\R}^3}\sup_{|k|\in K}|\varphi^\pm(x,k)|<c\ .
$$
By collecting these various results one has thus proved :

\begin{proposition}\label{proposition_stationary}
Assume that $\sigma>2$ and let $s \in (3/2,\sigma-1/2)$. Let  $f,g \in \H_s$ with $E_0([a,b])g=g$ for some $[a,b] \subset (0,\infty) \setminus \sigma_p(H)$. Then
$$
\big\langle f, W_\pm g\big \rangle_\H =
\big \langle f, (2\pi)^{-3/2} \int_{\R^3} \varphi^{\pm}(\cdot, k)  \hat g(k) \, \d k\big \rangle_{s,-s} \ .
$$
\end{proposition}


\section{Asymptotic limit for the wave operators}\label{secwave}
\setcounter{equation}{0}

In this section we study the behavior of the wave operators under the dilation group. A related study for Schr\"odinger operators in $\R^3$ is contained in \cite[Sec.~I]{AGHH}. As explained in the Introduction and as it will appear in the sequel, this study is related to the $0$-energy properties of $H$.

So, let us recall the action of the dilation group $\{U_\tau\}_{\tau \in \R}$ on any $f \in \H$, namely $[U_\tau f](x) = e^{3\tau/2}f(e^\tau x)$ for any $x \in \R^3$. Then, the following equalities hold for any fixed $\tau \in \R$ and any $f \in \D(H_0)$:
\begin{equation*}
U_{-\tau}H_0U_\tau f = e^\tau H_0 f \qquad {\rm and}
\qquad U_{-\tau}VU_\tau = V_\tau
\end{equation*}
with $V_\tau(x) = V(e^{-\tau}x)$ for all $x \in \R^3$. As a consequence of these relations one infers from the time-dependent expression for the wave operators the important relations
\begin{equation}\label{transfer}
U_{-\tau}\ W_\pm(H_0+V,H_0)\ U_\tau = W_\pm(H_0+e^{-\tau}V_\tau,H_0).
\end{equation}
For clarity, the dependence of $W_\pm$ on both self-adjoint operators used to define them is mentioned. Our aim in this section is to study the limits of the corresponding stationary expressions as $\tau \to - \infty$.

For that purpose, observe that for $z \in \C\setminus \R$ one has
$U_{-\tau}R_0(z)U_\tau = e^{-\tau}R_0(e^{-\tau}z)$. Furthermore, by setting
\begin{equation*}
v_\tau := |V(e^{-\tau}\cdot)|^{1/2} \qquad {\rm and} \qquad
u_\tau := |V(e^{-\tau}\cdot)|^{1/2} \;{\rm sgn}\big(V(e^{-\tau}\cdot)\big),
\end{equation*}
and by considering these functions as operators of multiplication ({\it i.e.}~$v_\tau \equiv v_\tau(X)$ and similarly for $u_\tau$) one also obtains
\begin{align*}
\big(1 + e^{-\tau} V_\tau  R_0(z) \big)^{-1} V_{\tau}
&=
v_\tau \big(1+e^{-\tau} u_\tau R_0(z) v_\tau\big)^{-1} u_\tau  \\
&=
 U_{-\tau}\;\! v_0 \big(1+ u_0 R_0(e^\tau z) v_0\big)^{-1}u_0 \;\!U_\tau.
\end{align*}
Thus, by the resolvent equation it follows that
\begin{align}\label{resoleqsym}
\nonumber &\big(H_0+e^{-\tau}V_\tau -z\big)^{-1} \\
\nonumber &  = R_0(z) -e^{-\tau}R_0(z)
\big( 1+e^{-\tau} V_\tau R_0(z) \big)^{-1} V_\tau R_0(z) \\
&= R_0(z)
- e^{-\tau}R_0(z) U_{-\tau}\;\! v_0 \big(
1+ u_0 R_0(e^\tau z) v_0\big)^{-1}u_0
\;\!U_\tau R_0(z).
\end{align}
From now on, for simplicity and in accordance with the notations introduced in Section \ref{sec0behavior}, we shall simply write $v$ for $v_0$ and $u$ for $u_0$.

Now, let us come back to the setting of Lemma \ref{lem:wosr1} but for the perturbation $e^{-\tau} V_\tau$ instead of $V$. We state in the next lemma alternative stationary expressions for the wave operators. In the statement, the parameter $\tau \in \R$ is fixed.

\begin{lemma}\label{lemalternative}
Let $\sigma>2$
 in Condition \eqref{condition1}, $s>1/2$ and  assume that $f,g$ belong to $\H_s$ with $E_0([a,b])g=g$ for some $[a,b] \subset (0,\infty)$ with $[a,b]\cap \sigma_p(H_0 + e^{-\tau} V_\tau) = \emptyset$.
Then one has
\begin{align}\label{eachfactor}
\nonumber &\big\langle f, \big(W_\pm(H_0+e^{-\tau} V_\tau,H_0)-1\big) g\big \rangle _\H \\
&= - e^{-\tau}\int_a^b \big\langle
B(e^{\tau}\lambda\pm i0)u
\;\!U_\tau R_0(\lambda\pm i 0)
f, v U_\tau E_0'(\lambda)g
\big\rangle_{\H} \,\d \lambda,
\end{align}
where $B(z) := \big(
1+ u R_0(z) v\big)^{-1}$.
\end{lemma}

\begin{proof}
Let $\varepsilon>0$ and $f,g$ as in the statement. It follows from \eqref{eqini}, \eqref{waveop+u11} and \eqref{resoleqsym} that
\begin{align*}
&\big\langle f, \big(W_\pm(H_0+e^{-\tau} V_\tau,H_0)-1\big) g\big \rangle _\H \\
&= - e^{-\tau} \lim_{\varepsilon \searrow 0}
\int_{-\infty}^{\infty} \big\langle
U_{-\tau}\;\! v B \big(e^{\tau}(\lambda\pm i\varepsilon)\big)u
\;\!U_\tau R_0(\lambda\pm i \varepsilon)
f, \delta_{\varepsilon}(H_0-\lambda)g
\big\rangle_{\H} \,\d \lambda \\
&= - e^{-\tau}\int_a^b \lim_{\varepsilon \searrow 0} \big\langle
B \big(e^{\tau}(\lambda\pm i\varepsilon)\big)u
\;\!U_\tau R_0(\lambda\pm i \varepsilon)
f, v U_\tau \delta_{\varepsilon}(H_0-\lambda)g
\big\rangle_{\H} \,\d \lambda.
\end{align*}
Note that the permutation of the integral and of the limit as well as the change in the endpoints of the integral are a consequence of \cite[Lem.~5.2.2]{Yaf}, as already mentioned in the proof of Lemma \ref{lem:wosr1}. Furthermore, since $U_\tau$ leaves $\H_{-s}$ invariant, it follows from the limiting absorption principle that both limits
$s-\lim_{\varepsilon\searrow 0} u U_\tau R_0(\lambda\pm i \varepsilon)f$
and
$s-\lim_{\varepsilon\searrow 0} v U_\tau \delta_{\varepsilon}(H_0-\lambda)g$ exist
and belong to $\H$.

Let us finally show that the limits $\lim_{\varepsilon \searrow 0}B \big(e^{\tau}(\lambda\pm i\varepsilon)\big)$ exist in norm for any $\lambda \in [a,b]$. For that purpose, it is sufficient to prove that
$-1 \not \in \sigma(u R_0(e^{\tau}\lambda \pm i0)v)$
if $\lambda \not\in \sigma_p(H_0 + e^{-\tau} V_\tau)$.
However, it follows from \cite[Lem.~4.7.8]{Yaf} and from our assumption on $V$ that $-1 \not \in \sigma(u R_0(e^{\tau}\lambda \pm i0)v)$ is equivalent to $e^\tau\lambda \not \in \sigma_p(H_0+V)$.
Then, from the equality
\begin{equation}  \label{egalspect}
\sigma_p(H_0 + e^{-\tau} V_\tau) = e^{-\tau}\sigma_p(H_0+V),
\end{equation}
one also infers that $e^\tau\lambda \not \in \sigma_p(H_0+V)\Longleftrightarrow \lambda\not\in \sigma_p(H_0 + e^{-\tau} V_\tau)$, which corresponds precisely to our assumption.
\end{proof}

Thus, we are now left in understanding separately the limits as $\tau \to -\infty$ of each factor in \eqref{eachfactor}. The study of two terms relies on the following two lemmas. These results are certainly well known, but we could not find an explicit reference for them.

\begin{lemma}\label{smallemma1}
Let $u,g \in L^2(\R^3)$ and $\varepsilon\geq 0$. Let $T_\varepsilon$ be the operator defined by the kernel $t_\varepsilon(x,y):= u(x)\;\!g(\varepsilon x-y)$ for almost every $x,y \in \R^3$. Then $T_\varepsilon$ is a Hilbert-Schmidt operator and converges to $T_0$ in the Hilbert-Schmidt norm as $\varepsilon\searrow 0$.
\end{lemma}

\begin{proof}
Since
\begin{equation*}
\|T_\varepsilon\|_{HS}^2 = \int_{\R^3}\int_{\R^3}|u(x)|^2 \;\!|g(\varepsilon x-y)|^2 \d x \;\!\d y=
\|u\|_{L^2(\R^3)}^2\;\!\|g\|_{L^2(\R^3)}^2,
\end{equation*}
it clearly follows that $T_\varepsilon$ is a Hilbert-Schmidt operator for any $\varepsilon\geq 0$.

For the convergence, let $K \subset \R^3$ be a compact set and let $K^\bot$ denote its complement. Then one has
\begin{align}
\nonumber \|T_\varepsilon-T_0\|_{HS}^2
&=\int_{\R^3}\int_{\R^3}|u(x)|^2\;\! |g(\varepsilon x-y)-g(-y)|^2 \d x \;\!\d y \\
\label{eqlem1} & \leq \|u\|_{L^2(K)}^2 \;\!\sup_{x \in K} \|g(\cdot -\varepsilon x)-g(\cdot)\|_{L^2(\R^3)}^2
+ 4 \|u\|_{L^2(K^\bot)}^2\;\! \|g\|_{L^2(\R^3)}^2.
\end{align}
By choosing a suitable set $K$ and then by taking the continuity of translations in $L^2(\R^3)$ into account, both terms in \eqref{eqlem1} can be made arbitrarily small for $\varepsilon$ small enough. This proves the statement.
\end{proof}

\begin{lemma}\label{smallemma2}
Let $u \in L^2(\R^3)$, $g \in L^1(\R^3)$ and $\varepsilon >0$. Let $T_\varepsilon$ be the operator defined by the kernel $t_\varepsilon(x,y):= u(x)\;\!g(\varepsilon x-y)$ for almost every $x,y \in \R^3$. Then $T_\varepsilon$ maps $L^\infty(\R^3)$ into $L^2(\R^3)$. Furthermore, for each $f \in L^\infty(\R^3)$, $T_\varepsilon f$ strongly converges to $u(\cdot)\int_{\R^3}g(-y)f(y)\d y$ in
 $L^2(\R^3)$ as $\varepsilon\searrow 0$.
\end{lemma}

\begin{proof}
For $f \in L^\infty(\R^3)$ observe first that
\begin{align*}
\|T_\varepsilon f\|_{L^2(\R^3)}^2
& = \int_{\R^3}\Big|u(x) \int_{\R^3}g(\varepsilon x-y)f(y)\d y\Big|^2 \d x \\
&\leq \int_{\R^3}|u(x)|^2\Big[\int_{\R^3}|g(\varepsilon x-y)f(y)|\d y\Big]^2 \d x \\
&\leq \|u\|_{L^2(\R^3)}^2\;\!\|f\|_{L^\infty(\R^3)}^2\;\!\|g\|_{L^1(\R^3)}^2.
\end{align*}

For the convergence, let $K \subset \R^3$ be a compact set and let $K^\bot$ denote its complement. Then one has
\begin{align}
\nonumber \big\|T_\varepsilon f- u(\cdot)
\int_{\R^3}g(-y)f(y)\d y\big\|^2_{L^2(\R^3)}
&=\int_{\R^3}|u(x)|^2\;\! \Big|\int_{\R^3}\big(g(\varepsilon x-y)-g(-y)\big)f(y)\d y\Big|^2 \d x \\
\nonumber & \leq \|u\|_{L^2(K)}^2 \;\!\|f\|_{L^\infty}^2 \;\!\sup_{x \in K} \|g(\cdot -\varepsilon x)-g(\cdot)\|_{L^1(\R^3)}^2 \\
\label{eqlem2}  & + 4 \|u\|_{L^2(K^\bot)}^2\;\! \;\!\|f\|_{L^\infty(\R^3)}^2\;\!\|g\|_{L^1(\R^3)}^2.
\end{align}
By choosing a suitable set $K$ and then by taking the continuity of translations in $L^1(\R^3)$ into account, both terms in \eqref{eqlem2} can be made arbitrarily small for $\varepsilon$ small enough. This proves the statement.
\end{proof}

By collecting these results, we can now analyse part of the terms in \eqref{eachfactor}. This study is contained in the next lemma.

\begin{lemma}\label{endfunctions}
Let us assume that $\sigma>3$ in Condition \eqref{condition1}, that $s>3/2$ and that $\lambda \in (0,\infty)$.
\begin{enumerate}
\item[(a)] For any $f \in \H_s\cap L^\infty(\R^3)$, the strong limits of
$e^{-3\tau/2} u U_\tau R_0(\lambda\pm i 0)f$ exist in $\H$ as $\tau \to -\infty$.
\item[(b)] For any $g\in \H_s$, the strong limit of $e^{-3\tau /2}v U_\tau E_0'(\lambda)g$ exists in $\H$ as $\tau \to -\infty$.
\end{enumerate}
\end{lemma}

\begin{proof}
(a) From the definition of $U_\tau$ and the explicit formulas \eqref{defg0},
\eqref{sum3} and \eqref{defg}, it follows that for almost every $x \in \R^3$:
\begin{align}\label{eqtau}
\nonumber &\big[e^{-3\tau/2} u  U_\tau R_0(\lambda\pm i 0)f\big](x) \\
\nonumber &= u(x) \int_{\R^3} g_0(e^\tau x-y) f(y) \d y
+ u(x) \int_{\R^3} m_\lambda(e^\tau x-y) f(y) \d y \\
&\quad + u(x) \int_{\R^3} k_\lambda^\pm(e^\tau x-y) f(y) \d y
\end{align}
In order to apply Lemmas \ref{smallemma1} and \ref{smallemma2} below, we need the assumption that
$\sigma >3$, which implies that $u  \in L^2(\R^3)$.

Since $g_0 = g_1 + g_2$ with $g_j \in L^j(\R^3)$
(as already used in the proof of Lemma \ref{lemG0})
and since $m_\lambda \in L^2(\R^3)$ by
\eqref{defg1} and \eqref{defg2},
it follows from Lemmas \ref{smallemma1} and \ref{smallemma2} that the first two terms on the r.h.s.~of \eqref{eqtau} admit a strong limit as $\tau \to -\infty$, or more precisely:
\begin{align*}
&s-\lim_{\tau \to -\infty} \Big[u(\cdot) \int_{\R^3} g_0(e^\tau \!\! \cdot-y) f(y) \d y
+ u(\cdot) \int_{\R^3} m_\lambda(e^\tau \!\! \cdot-y) f(y) \d y\Big] \\
&= u(\cdot) \int_{\R^3}\big(g_0(-y)+ m_\lambda(-y)\big)f(y) \d y.
\end{align*}
Note that the r.h.s.~is well defined since $f \in L^2(\R^3)\cap L^\infty(\R^3)$.

For the third term in \eqref{eqtau},
we decompose $k_{\lambda}^\pm$ into two terms:
\begin{equation*}
k_{\lambda}^\pm=k_{\lambda, 1}^\pm + k_{\lambda, 2}^\pm
:= \chi_{b(0,1)} \;\!k_{\lambda}^\pm +
 (1-\chi_{b(0,1)})\;\!k_{\lambda}^\pm
\end{equation*}
in the same way as $g_0$ in Lemma \ref{lemG0}.
Observe that $k_{\lambda,1}^\pm \in L^2(\R^3)$ and $k_{\lambda,2}^\pm \in L^{\infty}(\R^3)$.
By Lemma \ref{smallemma1} the operators corresponding to the kernel
$u(x) k_{\lambda,1}^\pm(e^\tau x -y)$ are Hilbert-Schmidt and converge in the Hilbert-Schmidt norm to the operators with kernel $u(x) k_{\lambda,1}^\pm(-y)$ as $\tau \to -\infty$.
Furthermore, since $f \in L^1(\R^3)$ and
since $k_{\lambda,2}^\pm \in L^{\infty}(\R^3)$,
Lemma \ref{smallemma2} shows that
\begin{align*}
s-\lim_{\tau \to -\infty} u(\cdot) \int_{\R^3} k_{\lambda,2}^\pm(e^\tau \!\!\cdot-y) f(y) \d y
&=
s-\lim_{\tau \to -\infty} u(\cdot) \int_{\R^3} f(e^\tau \!\!\cdot-y)  k_{\lambda,2}^\pm(y) \d y  \\
&=
u(\cdot)\int_{\R^3} f(-y) k_{\lambda,2}^\pm(y)  \d y  \\
&=
u(\cdot)\int_{\R^3} k_{\lambda,2}^\pm(-y) f(y) \d y.
\end{align*}

(b) It follows from \eqref{eqn:3q1} that for almost every $x \in \R^3$
\begin{equation*}
\big[e^{-3\tau/2} v \;\!U_\tau E_0^{\prime}(\lambda) g\big](x)
=
\frac{\lambda}{2\pi^2} v(x) \int_{\R^3} \frac{\sin(\lambda |e^\tau x-y|)}{|e^\tau x -y|}g(y) \d y.
\end{equation*}
Since $g \in L^1(\R^3)$
and the map $\R^3 \ni x \mapsto \frac{\sin(|x|)}{|x|}\in \R$ belong to
$L^{\infty}(\R^3)$,
one proves as above that
\begin{equation*}
s-\lim_{\tau \to -\infty} v(\cdot)
\int_{\R^3} \frac{\sin(\lambda |e^\tau \!\! \cdot -y|)}{|e^\tau \!\!\cdot -y|} g(y) \d y =
v(\cdot)\int_{\R^3} \frac{\sin(\lambda |y|)}{|y|} g(y) \d y.
\end{equation*}
\end{proof}

We are now left with the study of the asymptotic behaviors of
the operators $B(e^\tau \lambda \pm i 0)$ in \eqref{eachfactor}
as $\tau \to -\infty$.
In order to deal with the assumption
$[a, \, b] \cap \sigma_p(H_0 + e^{-\tau} V_\tau)= \emptyset$ of Lemma \ref{lemalternative}, let us observe that the equality \eqref{egalspect} implies that if $\lambda_0$ is a positive eigenvalue of $H_0+V$, then $e^{-\tau}\lambda_0$ is a positive eigenvalue of $H_0 + e^{-\tau} V_\tau$.
Now, by choosing $\tau$ close enough to $-\infty$, the value $e^{-\tau}\lambda_0$ can be made arbitrarily large. Thus, one infers that with the following implicit condition,
the mentioned assumption becomes manageable.

\begin{assumption}\label{Himplicit}
The value $0$ is not an accumulation point of positive eigenvalues for the operator $H_0+V$.
\end{assumption}

Obviously, this assumption is rather natural and a large class of perturbations $V$ should satisfy it. In Section \ref{sectionac} we provide sufficient conditions such that the spectrum of $H$ on $\R_+$ is purely absolutely continuous. However, the absence of accumulation of positive eigenvalues at $0$ is certainly verified under weaker assumptions.
Now, note that Assumption \ref{Himplicit} together with \eqref{egalspect} have an important consequence: for any $[a,b]\subset (0,\infty)$, there exists $\tau_{ab}\in \R$ such that for any $\tau \leq \tau_{ab}$, one has
\begin{equation}\label{condspect}
\sigma_p(H_0 + e^{-\tau} V_\tau) \cap [a,b] = e^{-\tau}\sigma_p(H_0+V)\cap [a,b] = \emptyset.
\end{equation}
In fact, for any $\tau \leq \tau_{ab}$ the even stronger statement
$\sigma_p(H_0 + e^{-\tau} V_\tau)\cap(0,b] = \emptyset$ holds.

For the time being, we shall impose an additional condition
(Assumption \ref{terrible} below) on the behavior of the $0$-energy threshold. It is not clear yet if this condition is necessary or even if it is always satisfied (see also Remark \ref{possibleextention} after Proposition \ref{lembehavior0}). So, let us assume that $\sigma>2$ in Condition \eqref{condition1} and denote by $\G_0$ the finite dimensional subspace of $\H$ spanned by the eigenvectors of the compact operator $uG_0v$ associated with the eigenvalue $-1$. The orthogonal projection on this subspace is simply denoted by $P$. In the Schr\"odinger case, this space corresponds to the set of $0$-energy eigenvectors and $0$-energy resonances. Our additional condition corresponds to the invertibility of a certain operator when restricted on $\G_0$. More precisely, let $Q_0$ be the operator
whose kernel is $\frac{1}{4\pi}|x-y|^{-1}$. Clearly, this operator corresponds to the resolvent of the Laplace operator at $0$-energy.

\begin{assumption}\label{terrible}
The operator $PuQ_0v\big|_{\G_0}: \G_0 \to \G_0$ is invertible.
\end{assumption}

Before proving the main result about the operator $B(e^\tau \lambda \pm i 0)$ in \eqref{eachfactor}, let us show that Assumptions \ref{Himplicit} and \ref{terrible} are generically satisfied.
Indeed, we prove in Lemma \ref{lembehavior1}
below that the condition $0 \not \in \sigma_p(H)$ implies that both Assumptions \ref{Himplicit} and \ref{terrible} hold. Then, since the operator $H$ rarely has
the $0$-energy eigenvalue
(see Propositions \ref{prp:sparse1} and \ref{prp:sparse2} in Section \ref{sec0behavior}),
it follows that the mentioned assumptions are almost always satisfied.

\begin{lemma}\label{lembehavior1}
Let $\sigma>2$ in Condition \eqref{condition1} and suppose that $0 \not \in \sigma_p(H)$.
Then both Assumptions \ref{Himplicit} and \ref{terrible} hold.
\end{lemma}

\begin{proof}
Under the same assumptions on $V$ and on $\sigma_p(H)$, it has already been shown in the proof of Theorem \ref{prp:0eigenvl} that $H$ has no eigenvalue in $[0,\lambda_0)$, for some $\lambda_0>0$. Thus, Assumption \ref{Himplicit} is satisfied. In addition, it has been proved in Lemma \ref{lem:0eigenvl-1} that if $0 \not \in \sigma_p(H)$ then $-1 \not \in \sigma(uG_0v)$, which immediately implies Assumption \ref{terrible} since the subspace $\G_0$ is then trivial.
\end{proof}

In addition to Lemma \ref{lembehavior1}, we would like to
mention that both Assumptions \ref{Himplicit} and \ref{terrible}
are verified for $H=H_0 + a V$ if $V$ satisfies Condition \eqref{condition1} with $\sigma>2$ and
$a \in \R$ is small enough.
Indeed, since $\sigma_p(H_0)=\emptyset$, this easily follows from Proposition \ref{prp:sparse1} and from  Lemma \ref{lembehavior1}.

\begin{proposition}\label{lembehavior0}
Let $\sigma>3$ in Condition \eqref{condition1} and suppose that Assumptions \ref{Himplicit} and \ref{terrible} hold.
Let $[a,b]\subset (0,\infty)$ and $\lambda \in [a,b]$. Then, there exists $\tau_{ab}\in \R$ such that for any $\tau \leq \tau_{ab}$, the operators $B(e^\tau\lambda \pm i 0)$ belong to $\B(\H)$ and the norm limits
\begin{equation}\label{thelimit}
\lim_{\tau \to -\infty} e^\tau B(e^\tau \lambda \pm i 0) \in \B(\H)
\end{equation}
exist.
\end{proposition}

\begin{proof}
It was already shown in the proof of Lemma \ref{lemalternative} that $1 + uR_0(e^\tau\lambda \pm i 0)v$ are invertible in $\B(\H)$ if
$\lambda \not \in \sigma_p(H_0+e^{-\tau}V_\tau)$. Furthermore, it follows from the above considerations that there exists $\tau_{ab}$ such that for $\tau \leq \tau_{ab}$ one has $\sigma_p(H_0 + e^{-\tau} V_\tau)\cap(0,b] = \emptyset$, which clearly prevents $\lambda$ from being an eigenvalue of $H_0+e^{-\tau}V_\tau$. Thus, $1 + uR_0(e^\tau\lambda \pm i 0)v$
are invertible in $\B(\H)$ and the inverses are by definition the operators $B(e^\tau\lambda \pm i 0)$.

Now, we already know from Lemma \ref{funclimit} that
$uR_0(e^\tau\lambda \pm i 0)v$
converge in norm to $uG_0v$ as $\tau \to -\infty$.
However, depending if $-1$ belongs to the spectrum of $uG_0v$ or not,
the behaviors of $B(e^\tau\lambda \pm i 0)$ as $\tau \to -\infty$ change drastically.
Clearly, if $-1 \not \in \sigma (uG_0v)$, then
$B(e^\tau\lambda \pm i 0)$ converge in norm to $(1+uG_0v)^{-1}$ as $\tau \to -\infty$,
and in that case the limits in \eqref{thelimit} are equal to $0$. But if $-1 \in \sigma (uG_0v)$, a more refined work is necessary. The rest of the proof is divided into several steps.

(a) We first derive better approximations for the operators $K_{e^\tau\lambda}^\pm$ and $M_{e^\tau\lambda}$. For simplicity, let us set $\varepsilon :=e^\tau\lambda$ and observe that
\begin{equation*}
k_\varepsilon^\pm(x)= \frac{\varepsilon}{2\pi}\frac{1}{|x|} + \frac{\varepsilon}{2\pi}\frac{e^{\pm i \varepsilon |x|}-1}{|x|}
= \frac{\varepsilon}{2\pi}\frac{1}{|x|} \pm i \frac{\varepsilon^2}{2\pi} \int_0^1 e^{\pm i s \varepsilon |x|}\;\! \d s.
\end{equation*}
It follows that
\begin{equation*}
[uK_\varepsilon^\pm v](x,y)
=
\frac{\varepsilon}{2\pi} u(x)\frac{1}{|x-y|}v(y)
\pm i \frac{\varepsilon^2}{2\pi} u(x)
\Big[\int_0^1 e^{\pm i s \varepsilon |x-y|}\;\! \d s \Big]v(y).
\end{equation*}
By setting $Q_0$ for the operator with kernel $\frac{1}{4\pi}|x-y|^{-1}$ the previous equality reads
\begin{equation*}
uK_\varepsilon^\pm v
=
2 \varepsilon \;\! u Q_0 v \pm  \varepsilon^2 \;\!B_\varepsilon^\pm
\end{equation*}
where $B_\varepsilon^\pm$ are Hilbert-Schmidt operators with Hilbert-Schmidt norms bounded by a constant independent of $\varepsilon$.

For the operator $M_\varepsilon$, let us observe that
\begin{equation*}
m_\varepsilon(x)= -\frac{\varepsilon}{4\pi|x|}
+\frac{\varepsilon^2}{2\pi^2}\cdot \frac{1}{\varepsilon|x|}
\big(
\sin(\varepsilon|x|)\;\ci(\varepsilon|x|)+ \cos(\varepsilon|x|)\;\si(\varepsilon|x|)+ \frac{\pi}{2}\big).
\end{equation*}
Note now that the function $\rho \mapsto \frac{1}{\rho}\big(
\cos(\rho)\;\si(\rho)+\frac{\pi}{2}\big)$ is bounded on $(0,\infty)$. On the other hand, the function
$\rho\mapsto \frac{1}{\rho}\sin(\rho)\;\ci(\rho)$ is bounded in the  neighbourhood of $+\infty$ but only the map
$\rho\mapsto \frac{1}{\rho \ln \rho} \sin(\rho)\;\ci(\rho)$ is
bounded in the neighbourhood of $0$.
Taking account of these facts,
we introduce a cut-off function $\chi$, which is continuous on $[0, \, \infty)$, with
\ $\chi(\rho)=1$ for $\rho\in [0,\, 1/2]$, \ $\chi(\rho)=0$ for $\rho \in [3/4, \, \infty)$ and $0 \le \chi \le 1$,
and accordingly decompose $m_\varepsilon(x)$
in the following manner:
\begin{align}\label{devm}
& m_\varepsilon(x) =  -\frac{\varepsilon}{4\pi|x|}
+\varepsilon^2 \;\!\ell(\varepsilon |x|) + \varepsilon^2 \;\!\ln (\varepsilon |x|)  \;\!n(\varepsilon |x|),  \\
& \ell(\rho):= \frac{1}{2\pi^2}\Big[
\frac{1-\chi(\rho)}{\rho} \sin(\rho)\;\ci(\rho) +
\frac{1}{\rho}\big(
\cos(\rho)\;\si(\rho)+\frac{\pi}{2}\big) \Big],  \nonumber  \\
& n(\rho):=
\frac{1}{2\pi^2}\Big[
\frac{\chi(\rho)}{\rho \ln \rho} \sin(\rho)\;\ci(\rho)
\Big].    \nonumber
\end{align}
Let $D_\varepsilon$ denote the operator with kernel
$u(x)\ln (\varepsilon |x-y|)  \;\!n(\varepsilon |x-y|)v(y)$.
We then observe that for any $\gamma>0$
\begin{align*}
\|D_\varepsilon\|_{HS}^2
&= \int_{\R^6} \big|u(x)\ln (\varepsilon |x-y|)  \;\!n(\varepsilon |x-y|)v(y)\big|^2 \d x\;\!\d y \\
&\leq {\rm Const.}~\varepsilon^{-2\gamma}
\int_{\R^6} \Big|u(x)\frac{1}{|x-y|^{\gamma}} v(y)\Big|^2 \d x\;\!\d y \\
&\leq {\rm Const.}~\varepsilon^{-2\gamma}
\int_{\R^6} \langle x \rangle^{-\sigma}\frac{1}{|x-y|^{2\gamma}} \langle y \rangle^{-\sigma}
\d x\;\!\d y \\
&\leq \varepsilon^{-2\gamma} \ {\rm Const}(\gamma,\sigma).
\end{align*}
For the last equality, one has used estimates for convolution operator obtained in \cite[Lem.~11.1]{U}.

By collecting these results and by fixing $\gamma=1/2$ one has thus obtained that
\begin{gather}\label{asympt}
\begin{split}
u R_0(\varepsilon \pm i 0) v
&=
u(G_0+ K_\varepsilon^\pm +M_\varepsilon)v  \\
&=
u G_0 v +
\varepsilon \;\! u Q_0 v + \varepsilon^{3/2} D_\varepsilon +  \varepsilon^2 C_\varepsilon^\pm
\end{split}
\end{gather}
where $C_\varepsilon^\pm$ and $D_\varepsilon$ are Hilbert-Schmidt operators with Hilbert-Schmidt norms bounded by constants independent of $\varepsilon$.

(b) For the second step of the proof, we can rely on results obtained in \cite[Sec.~I.1.2]{AGHH} for the Schr\"odinger case. Indeed, the single difference between both contexts is the definition of the operator $G_0$, but the rest of the analysis can be mimicked.
Then, based on \cite[Chap.~III.6.5]{Kato}, it has been proved in \cite[Sec.~I.1.2]{AGHH} that for any $z \in \C \setminus \{0\}$ with $|z|$ small enough, the following norm convergent expansion holds:
\begin{equation}\label{series}
(1+ u G_0v+z)^{-1} = z^{-1}P + \sum_{m=0}^\infty (-z)^m T^{m+1},
\end{equation}
where $P$ is the projection onto the eigenspace of $uG_0v$ associated with the eigenvalue $-1$
and $T \in \B(\H)$.

(c) Let us now come to the main part of the proof. By taking the estimates \eqref{asympt} and \eqref{series} into account, observe that for $\tau \leq \tau_{ab}$ one has
\begin{align*}
&e^\tau B(e^\tau \lambda \pm i 0) \\
&= e^\tau \big(1+uR_0(e^\tau \lambda \pm i0)v\big)^{-1} \\
&= e^\tau \big(1+uG_0v + e^\tau \lambda uQ_0v + o(e^\tau)\big)^{-1} \\
&= e^\tau \Big( \big(1+e^\tau + uG_0v \big)
\big[1+ e^\tau \big(1+e^\tau + uG_0v\big)^{-1} \big(\lambda uQ_0v -1 + o(1)\big)\big]\Big)^{-1} \\
&=  \Big(1+ \big(P+O(e^\tau)\big) \big(\lambda uQ_0v -1 + o(1)\big)\Big)^{-1}\big(P+O(e^\tau)\big) \\
&= \big(1+ P (\lambda uQ_0v -1) + o(1)\big)^{-1}\big(P+O(e^\tau)\big),
\end{align*}
where the symbols $o(e^{j\tau})$ and $O(e^{j\tau})$ mean respectively that $\lim_{\tau \to -\infty}e^{-j\tau}\|o(e^{j\tau})\|_{\B(\H)}= 0$ and
$e^{-j\tau}\|O(e^{j\tau})\|_{\B(\H)} \in L^\infty(-\infty,\tau_{ab})$ for $j \in \{0,1\}$.
Thus, if the operator $1+ P (\lambda uQ_0v -1)$ is invertible with a bounded inverse,
then the norm limit $\tau \to -\infty$ can be performed in the previous expression and one obtains
\begin{equation*}
\lim_{\tau \to -\infty}  e^\tau B(e^\tau \lambda \pm i 0) =
\big(1+ P (\lambda uQ_0v -1)\big)^{-1}P.
\end{equation*}
Therefore, the final step in the proof consists in studying
the operator $1+ P(\lambda uQ_0v -1)$.

(d) Since $P$ is a projection, one observes that
the invertibility of $1+ P(\lambda uQ_0v -1)$ holds in $\B(\H)$
if the condition of  Assumption \ref{terrible} is satisfied.
\end{proof}

\begin{remark}\label{possibleextention}
It is possible to avoid assuming Assumption \ref{terrible} by still improving part of the previous proof. Indeed, by further developing the term $m_\varepsilon$ in \eqref{devm}, then by working more carefully and by considering another expression of $e^\tau$ in front of the term $B(e^\tau \lambda \pm i0)$, a better analysis
in the line of \cite{JN}
could be performed without the Assumption \ref{terrible}. For the time being, we do not carry out this computation. In comparison, let us mention that in the Schr\"odinger case, a similar study has been avoided in \cite{AGHH} by inserting an additional real-analytic function of $\tau$ just before $V$ and by adding sufficient conditions on this function. Thanks to this trick, the authors avoid a condition similar to our Assumption \ref{terrible} but it also prevents them from considering all the possible situations.
\end{remark}

Summing all the results obtained so far, one can readily prove the following statement.

\begin{proposition}\label{weakl}
Let us assume that $\sigma>3$ in Condition \eqref{condition1},
and suppose that Assumptions \ref{Himplicit} and \ref{terrible} hold. For $s>3/2$, let $f\in \H_s\cap L^\infty(\R^3)$ and $g \in \H_s$ with  $E_0([a,b])g=g$ for some $[a,b] \subset (0,\infty)$.
Then the limits \begin{equation*}
\lim_{\tau \to -\infty}\big\langle f, \big(U_{-\tau} W_\pm(H,H_0)U_\tau -1\big) g\big \rangle _\H = 0
\end{equation*}
hold.
\end{proposition}

\begin{proof}
It is clear from
\eqref{transfer}, \eqref{condspect} and Lemma \ref{lemalternative}
that
there exists $\tau_{ab}\in \R$ such that for any $\tau \leq \tau_{ab}$ the stationary representations \eqref{eachfactor} hold.
Then, let us observe that
\begin{align*}
&\big\langle f, \big(W_\pm(H_0+e^{-\tau} V_\tau,H_0)-1\big) g\big \rangle _\H \\
&= - e^{-\tau}\int_a^b \big\langle
B(e^{\tau}\lambda\pm i0)u
\;\!U_\tau R_0(\lambda\pm i 0)
f, v U_\tau E_0'(\lambda)g
\big\rangle_{\H} \,\d \lambda \\
&= - e^{\tau}\int_a^b \big\langle
\big[e^\tau B(e^{\tau}\lambda\pm i0)\big]
e^{-3\tau/2}u \;\!U_\tau R_0(\lambda\pm i 0)
f, e^{-3\tau/2}v U_\tau E_0'(\lambda)g
\big\rangle_{\H} \,\d \lambda\ .
\end{align*}
It then follows from Lemma \ref{endfunctions} and Proposition \ref{lembehavior0} that
\begin{equation*}
\lim_{\tau \to -\infty} e^{\tau} \big\langle
\big[e^\tau B(e^{\tau}\lambda\pm i0)\big]
e^{-3\tau/2}u \;\!U_\tau R_0(\lambda\pm i 0)
f, e^{-3\tau/2}v U_\tau E_0'(\lambda)g
\big\rangle_{\H} = 0.
\end{equation*}
Finally, the permutation of the integral and of the limit is easily obtained by an application of the Lebesgue's dominated convergence theorem.
\end{proof}

\begin{theorem}\label{mainwave}
Let us assume that $\sigma>3$ in Condition \eqref{condition1},
and suppose that Assumptions \ref{Himplicit} and \ref{terrible} hold.
Then, the following limits hold:
\begin{equation*}
s-\lim_{\tau \to -\infty} U_{-\tau} W_\pm(H,H_0)U_\tau =1.
\end{equation*}
\end{theorem}

\begin{proof}
By density, it is sufficient to show that $\lim_{\tau \to -\infty}\|(U_{-\tau} W_\pm(H,H_0)U_\tau -1)f\|_\H=0$ for any
$f \in \H$ with $\hat f \in C_c^\infty(\R^3\setminus\{0\})$.
Observe first that such $f$ satisfies all conditions imposed on $f$ and $g$ in the statement of Proposition \ref{weakl}.
Then, let us write $W_\pm(\tau)$ for the operator $U_{-\tau} W_\pm(H,H_0)U_\tau$ and compute
\begin{equation}\label{strong}
\big\|\big(W_\pm(\tau) -1\big) f\big \|_\H^2
= -\langle W_\pm(\tau)f,f\rangle_\H - \langle f, W_\pm(\tau)f\rangle_\H
+ \|W_\pm(\tau)f\|^2_\H + \|f\|^2_\H.
\end{equation}
By Proposition \ref{weakl} the first two terms converge to $-\|f\|^2_\H$ as $\tau \to -\infty$. In addition, observe that
$$
\|W_\pm(\tau)f\|^2_\H = \|W_\pm(H,H_0)U_\tau f\|^2_\H = \|U_\tau f\|^2_\H = \|f\|^2_\H
$$
because $W_\pm(H,H_0)$ are isometries. Thus, the expressions on the l.h.s.~of \eqref{strong} converge to $0$ as $\tau \to -\infty$.
\end{proof}

The previous result has also important consequences on the scattering operator $S$ as we shall show in the next section.


\section{Asymptotic limit for the scattering operator}\label{secS}
\setcounter{equation}{0}

Let us first recall that the scattering operator $S$ is defined by the product $W_+^*W_-$ and is a unitary operator. Then, an immediate consequence of Theorem \ref{mainwave} reads as follows:

\begin{corollary}\label{corol1}
Let us assume that $\sigma>3$ in Condition \eqref{condition1},
and suppose that Assumptions \ref{Himplicit} and \ref{terrible} hold.
Then the following limit holds:
\begin{equation}\label{limitS}
s-\lim_{\tau \to -\infty} U_{-\tau} S U_\tau =1.
\end{equation}
\end{corollary}

\begin{proof}
Let us set
$S_{\tau}$
for $U_{-\tau}SU_\tau$ and recall the notation $W_\pm(\tau):=U_{-\tau}W_\pm U_\tau$ introduced in the previous proof. Then, for any $f \in \H$ one deduces from Theorem \ref{mainwave} that
\begin{align*}
\lim_{\tau \to -\infty} \langle f, S_{\tau}f\rangle_\H
=\lim_{\tau \to -\infty} \langle W_+(\tau)f,W_-(\tau)f\rangle_\H
= \|f\|_\H^2.
\end{align*}
With a trick similar to the one already used in the proof of that theorem, one then deduces that $\lim_{\tau \to -\infty}\|(S_{\tau}-1)f\|^2_\H=0$.
\end{proof}

Let us also look at the consequence of the previous results on the scattering matrices. For that purpose, let $\HS:= L^2(\S^2)$ and $\FF_0 :\H\to L^2(\R_+, \d\lambda;\HS)=:\Hrond$ be the unitary transformation which diagonalizes the operator $H_0$, namely
$[\FF_0 H_0 f](\lambda)=\lambda [\FF_0 f](\lambda)$
for any $f$ belonging to $\H^1$ and for almost every $\lambda \in \R_+$. For the relativistic Schr\"odinger operator, the expression for $\FF_0$ is very simple, more precisely for any $f \in \SS(\R^3)$ one has
$\big[[\FF_0 f](\lambda)\big](\omega):=\lambda [\FF f](\lambda \omega)$ for any $\lambda >0$ and $\omega \in \S^2$.

Now, it is well known that $S$ is diagonal in the spectral representation of $H_0$, or in other words
that $\FF_0 S \FF_0^* = S(\Lambda)$,
where $S(\Lambda)$
denotes an operator of multiplication on $\R_+$ by an essentially bounded function with values
in $\B(\HS)$. More precisely, for any $\varphi \in \Hrond$ and $\lambda \in \R_+$ the action of
$S(\Lambda)$ reads
$[S(\Lambda) \varphi](\lambda) = S(\lambda) \varphi(\lambda) \in \HS$
and $S(\lambda) \in \B(\HS)$ is called the scattering matrix at energy $\lambda$.
Then, by taking into account this relation as well as the well known equality
$\FF U_\tau \FF^* = U_{-\tau}$,
one infers that
$[\FF_0 U_{\tau} \FF_0^* \varphi](\lambda)
= e^{-\tau/2} \varphi(e^{-\tau} \lambda)$
for any
$\varphi \in \Hrond$,
and hence obtains that
$[\FF_0  S_\tau \FF_0^* \varphi](\lambda) = S(e^\tau \lambda) \varphi(\lambda)$.
By introducing the notation
$[S_{\tau}(\Lambda) \varphi](\lambda):= S(e^{\tau}\lambda) \varphi(\lambda)$, we infers that
$\FF_0  S_\tau \FF_0^* = S_{\tau}(\Lambda)$.
Note that $S_0=S$ and $S_0(\Lambda) = S(\Lambda)$.

In that setting, relation \eqref{limitS} reads as follows: For any $\varphi \in \Hrond$, one has
\begin{equation*}
\lim_{\tau \to -\infty} \|
S_{\tau}(\Lambda)\varphi - \varphi
\|_{\Hrond}=0.
\end{equation*}
However, this relation is strictly weaker than the uniform limit $u-\lim_{\lambda \searrow 0} S(\lambda) = 1$, which has been mentioned in the Introduction. In order to obtain the latter result, we shall borrow in the proof of the next statement a usual stationary representation of the scattering matrix.

\begin{theorem}
Let us assume that $\sigma>3$ in Condition \eqref{condition1},
and suppose that Assumptions \ref{Himplicit} and \ref{terrible} hold.
Then $u-\lim_{\lambda \searrow 0}S(\lambda) = 1$ in $\B(\HS)$.
\end{theorem}

\begin{proof}
For any $\lambda \in \R_+$, let us first introduce the operator $\FF_0(\lambda)$ defined on $f \in \SS(\R^3)$ by the relation $\FF_0(\lambda)f = [\FF_0 f](\lambda)\in \HS$. By analogy to the Schr\"odinger case, it is easily shown that this operator extends continuously to an element of $\B(\H_s,\HS)$ for any $s>1/2$.
Furthermore, by mimicking the approach presented in \cite[Sec.~5]{JK} an asymptotic expansion for $\FF_0(\lambda)$ as $\lambda \searrow 0$ can also be derived.
More precisely, one readily obtains that
$\FF_0(\lambda) =  \lambda \gamma_0
+ o(\lambda)$ in $\B(\H_s,\HS)$ for any $s>3/2$,
where $[\gamma_0 f](\omega)= \hat f (0)$.

Then, the following representation of the scattering matrix holds (see for example \cite[Sec.~2.8 \& 5.7]{Yaf} or \cite[Sec.~5]{JK}) :
\begin{equation*}
S(\lambda) = 1 -2\pi i \FF_0(\lambda) V \big(1+R_0(\lambda +i0) V\big)^{-1} \FF_0(\lambda)^*
\end{equation*}
In addition, by taking the following relations into account
\begin{equation*}
\big[V\big(1+R_0(z)V\big)^{-1}\big]^* =  \big(1+VR_0(\overline z)\big)^{-1}V
=  v \big(1+uR_0(\overline z)v\big)^{-1}u
\end{equation*}
one infers the useful relation
$$
S(\lambda) = 1 -2\pi i \FF_0(\lambda)
u \;\!B(\lambda -i0)^*\;\!v\FF_0(\lambda)^*
$$
where $B(z)$ was introduced in the statement of Lemma
\ref{lemalternative}.
In addition, recall from Proposition \ref{lembehavior0} that the norm limit $\lim_{\lambda \searrow 0} \lambda B(\lambda -i0)$ exists in $\B(\H)$. Thus, by taking into account the already mentioned properties of $\FF_0(\lambda)$ when $\lambda \searrow 0$, one directly deduces the statement of the theorem.
\end{proof}


\section{Absolute continuity of the spectrum on $\boldsymbol{[0,\infty)}$}\label{sectionac}
\setcounter{equation}{0}

The non-existence of embedded eigenvalues should certainly deserve more attention for the present model. However, since investigations on this question for Schr\"odinger operators always involve a rather heavy machinery, we do not expect that this question can be easily solved for the present relativistic model. On the other hand, by assuming stronger conditions on $V$, one can deduce from an abstract argument that the spectrum of $H$ on $\R_+$ is purely absolutely continuous. We clearly suspect that the following assumptions on $V$ are much too strong both for the non-existence of positive eigenvalues and for the absolute continuity of the spectrum on $\R_+$. But since the argument is rather simple, we have decided to present it for completeness. The proof is based on an abstract result obtained in \cite{Ri06}

Before going into the details of the application of \cite[Thm.~1]{Ri06}, let us recall one definition on regularity of operators with respect to $C_0$-groups. Let $\T_1,\T_2$ be two Banach spaces endowed with two $C_0$-groups $\{U^1_\tau\}_{\tau \in \R}, \{U^2_\tau\}_{\tau \in \R}$ of generators $A_1,A_2$, respectively. One says that an element $B \in \B(\T_1,\T_2)$ belongs to $C^1(A_1,A_2;\T_1,\T_2)$ if the map
$$
\R\ni \tau \mapsto U^2_{-\tau} \;\!B \;\!U^1_\tau \in \B(\T_1,\T_2)
$$
is strongly differentiable.

Now, recall that the dilation group has been introduced in Section \ref{secwave}. It is known that this group defines $C_0$-groups in all weighted Sobolev spaces $\H^t_s$, for $s,t \in \R$. Note that these groups are defined either by restrictions or by duality arguments, and that we keep the same notation $\{U_\tau\}_{\tau \in \R}$ for these groups in each of these spaces. Their generators are all denoted by $A$. Furthermore, the relation $U_{-\tau}H_0U_\tau = e^{\tau}H_0$ clearly holds in $\B(\H^1,\H)$ for all $\tau \in \R$.
As a consequence, the operator $H_0$ belongs to $C^1(A,A;\H^1,\H)\equiv C^1(A;\H^1,\H)$.

Let us now add the potential $V$. In the sequel, we assume that $V\in C^2_b(\R^3)$, which means that the potential, its first order derivatives as well as its second order derivatives are continuous and bounded.
We also assume that the function $\widetilde V$, defined by
$\widetilde V(x)= x\cdot [\nabla V](x)$ for all $x\in \R^3$, is a bounded function.
Since $U_{-\tau}VU_\tau$ is the operator of multiplication by the function $V_\tau$ defined by $V_\tau(x)=V(e^{-\tau}x)$ for any $x \in \R^3$, one easily observes that $V \in C^1(A; \H,\H)\equiv C^1(A;\H)$, and therefore $V \in C^1(A;\H^1,\H)$.
As a consequence, one deduces that $H$ belongs to $C^1(A;\H^1,\H)$ and the following equalities hold in $\B(\H^1,\H)$:
$$
\frac{\d}{\d \tau}\big(U_{-\tau}HU_\tau\big)\big|_{\tau=0} =
[iH,A]= H_0 - \widetilde V.
$$

For the application of \cite[Thm.~1]{Ri06}, one needs to impose a positivity condition as well as further decrease conditions. For that purpose, let us first recall Kato's inequality:
$H_0 \geq  2 {\pi}^{-1}|X|^{-1}$
({\it cf}.~\cite[Thm.~2.2.4]{BE},  \cite[p. 307]{Kato}).
Then, our positivity assumption takes the following form : there exist two constants $c_1, c_2\in [0,1)$ with $c_1+c_2<1$ such that
\begin{equation}\label{c1}
M:=2 {\pi}^{-1}c_2  \frac{1}{|X|}-c_1 V-\widetilde V>0.
\end{equation}
In other words, $M$ is the operator of multiplication by the non-negative function
$x \mapsto M(x):=2 {\pi}^{-1} c_2  \frac{1}{|x|}-c_1 V(x)-\widetilde V(x)$.
One infers from this inequality that the operator $T$, defined on $\H^{1}$ by $T:=-c_1H + [iH,A]$ satisfies
\begin{equation*}
T= (1-c_1)H_0-c_1 V -\widetilde V \geq (1-c_1-c_2)H_0 + M
>0.
\end{equation*}
One also gets the inequalities $T\geq M$, $T\geq (1-c_1-c_2)H_0$ and
$T\geq 2{\pi}^{-1}(1-c_1-c_2)|X|^{-1}$.

For the decrease conditions, let us assume that for all $x \in \R^3$:
\begin{equation}\label{c2}
|x\cdot [\nabla V](x)| \leq {\rm Const.}~\langle x\rangle^{-1}
\qquad \hbox{ and }\qquad
\big|x\cdot \nabla[(x\cdot \nabla)V](x)\big|\leq {\rm Const.}~\langle x\rangle^{-1}.
\end{equation}
Since
$H_0 \geq  2{\pi}^{-1}|X|^{-1}\geq 2{\pi}^{-1} \langle X\rangle^{-1}$,
one then infers that there exists a constant $c$ large enough such that the following inequalities hold:
\begin{align}
\label{l1} -c T\leq [iH,A]&\leq cT, \\
\label{l2} -c T\leq \big[i[iH,A],A\big]&\leq cT, \\
\label{l3} -c T\leq [iT,A]&\leq cT.
\end{align}
With these inequalities at hand, one can now prove:

\begin{proposition}
Assume that $V\in C^2_b(\R^3)$ and that the conditions contained in \eqref{c2} are satisfied. Assumed in addition that there exist two constants $c_1, c_2\in [0,1)$ with $c_1+c_2<1$ such that the condition \eqref{c1} is verified. Then, the operator $H$ has purely absolutely continuous spectrum on $[0,\infty)$.
\end{proposition}

\begin{proof}
The proof consists in checking that the abstract conditions of \cite[Thm.~1]{Ri06} are satisfied.
As already noticed before the statement of the proposition, one clearly has that $H$ belongs to $C^1(A;\H^1,\H)$ and that the operator $T=-c_1H + [iH,A]$ satisfies $T>0$ on $\H^1$. In addition, the operator $[iH,A]=H_0-\widetilde V$ is bounded from below. Thus both conditions contained in \cite[Eq.~(2)]{Ri06} are satisfied.

Now, let us keep writing $[iH,A]$ and $T$ for the continuous extensions of these operators to elements of $\B(\H^{1/2},\H^{-1/2})$. It then follows from \eqref{l1} that for all $f \in \H^{1/2}$ one has
\begin{equation}\label{c3}
\big|\langle f,[iH,A]f\rangle_{1/2,-1/2}\big|\leq  c\langle f,Tf\rangle_{1/2,-1/2}.
\end{equation}
Thus, if $\T$ denotes the completion of $\H^{1/2}$ with the norm $\|f\|_\T:= \langle f,Tf\rangle^{1/2}_{1/2,-1/2}$, it follows from \eqref{c3} that $[iH,A]$ extends to an element of $\B(\T,\T^*)$, where $\T^*$ denotes the adjoint space of $\T$. Note that relation \eqref{l2} leads to a similar conclusion for the operator
$\big[i[iH,A],A\big]$.

We finally check that $\{U_\tau\}_{\tau\in \R}$ extends
to a $C_0$-group in $\T$.
This easily reduces to the proof that
$\|U_\tau f\|_\T \leq c(\tau) \|f\|_\T$ for all $f \in \H^{1/2}$ and $\tau \in \R$.
By  \eqref{l3} one has~:
\begin{equation*}
\|U_\tau f\|^2_\T \ = \ \langle f, Tf\rangle +
\int_0^\tau \langle U_t f, [iT,A]U_t f\rangle \;\!\d t
\ \leq \ \|f\|^2_\T + c \;\!\Big|\int_0^\tau \|U_t f\|_\T^2 \; \d t\Big|\ .
\end{equation*}
The function $(0,\tau) \ni t \mapsto \|U_t f\|_\T^2 \in \R$ is bounded
(since $\H^{1/2} \hookrightarrow \T$), and hence by a simple form of the
Gronwall Lemma, we get the inequality $\|U_\tau f\|_\T \leq e^{\frac{c}{2}|\tau|}
\|f\|_\T$.
Thus $\{U_\tau\}_{\tau \in \R}$ extends to a $C_0$-group in $\T$, and by duality
$\{U_\tau\}_{\tau\in \R}$ also defines a $C_0$-group in $\T^*$.
This finishes the proof that $[iH,A]$ extends to
an element of $C^1(A;\T,\T^*)$.
All hypotheses of \cite[Thm.~1]{Ri06} have been checked,
and the statement follows from this theorem and from its corollary.
\end{proof}


\section{Appendix}
\setcounter{equation}{0}

In this appendix, we derive an explicit expression for the action of the unitary group generated by $H_0$. Apparently, such formula was not exhibited before.

For that purpose, let us consider $f\in C^\infty_c(\R^3)$,
$g \in \SS$ with $\hat g\in C^\infty_c(\R^3)$
and for $z \in \C$ one sets
$$
\zeta_{\pm}(z):=\int_{\R^3} e^{\pm iz|k|}\hat f(k) \;\!\overline{\hat g(k)}\;\!\d k.
$$
Clearly,
$\zeta_{\pm}$ are
entire functions on $\C$ and
one has
$\zeta_{\pm}(\mp t) = \big\langle e^{-itH_0}f,g\big\rangle$
for any $t \in \R$. On the other hand, one also has for any $t>0$
\begin{align*}
\zeta_\pm(\pm it)
= & \int_{\R^3}e^{-t|k|} \hat f(k) \;\!\overline{\hat g(k)}\;\!\d k \\
= & \big\langle e^{-tH_0}f,g\big\rangle \\
= & \int_{\R^3}\Big\{
\int_{\R^3}\frac{t}{\pi^2 (|x-y|^2+t^2)^2} \;\!f(y)\;\!\d y \Big\} \overline{g(x)}\;\!\d x,
\end{align*}
where the explicit form of the semi-group is borrowed from \cite[Eq.~(2.1)]{U}.
Now, by setting
\begin{equation*}
\eta_\pm(z)
:= \int_{\R^3}\Big\{
\int_{\R^3}\frac{\mp iz}
{\pi^2 (|x-y|^2-z^2)^2} \;\!f(y)\;\!\d y \Big\} \overline{g(x)}\;\!\d x
\end{equation*}
one easily observes that the maps
$\eta_\pm$
are holomorphic on
$\C_\pm$.
Furthermore, the equalities
$\zeta_\pm(\pm it)=\eta_\pm(\pm it)$
hold for any $t>0$. By analytic continuation, it follows that the functions
$\zeta_\pm$
and $\eta_\pm$ are equal on $\C_\pm$, respectively.

And as a consequence, one infers that for each fixed $t > 0 $ one has
\begin{align*}
\big\langle e^{-itH_0}f,g\big\rangle = &
\zeta_-(t) =
\lim_{\varepsilon \searrow 0}\zeta_-(t-i\varepsilon) = \lim_{\varepsilon \searrow 0}\eta_-(t-i\varepsilon) \\
= & \lim_{\varepsilon \searrow 0} \int_{\R^3}\Big\{
\int_{\R^3}\frac{it + \varepsilon}{\pi^2 (|x-y|+t-i\varepsilon)^2(|x-y|-t+i\varepsilon)^2} \;\!f(y)\;\!\d y \Big\} \overline{g(x)}\;\!\d x
\end{align*}
which formally reads
\begin{equation*}
\big\langle e^{-itH_0}f,g\big\rangle =
\int_{\R^3}\Big\{
\int_{\R^3}\frac{it}{\pi^2 (|x-y|+t)^2(|x-y|-t+i 0)^2} \;\!f(y)\;\!\d y \Big\} \overline{g(x)}\;\!\d x
\end{equation*}
where the distributions $s\mapsto \frac{1}{(s\pm i0)^2}$ are for example defined in \cite[Sec.~3.2]{Hoer}.
On the other hand, one infers for each fixed $t<0$ that
\begin{align*}
\big\langle e^{-itH_0}f,g\big\rangle = &
\zeta_+(-t) =
\lim_{\varepsilon \searrow 0}\zeta_+(|t|+i\varepsilon) = \lim_{\varepsilon \searrow 0}\eta_+(|t|+i\varepsilon) \\
= & \lim_{\varepsilon \searrow 0} \int_{\R^3}\Big\{
\int_{\R^3}\frac{it+\varepsilon}{\pi^2 (|x-y|-t+i\varepsilon)^2(|x-y|+t-i\varepsilon)^2} \;\!f(y)\;\!\d y \Big\} \overline{g(x)}\;\!\d x
\end{align*}
which formally reads
\begin{equation*}
\big\langle e^{-itH_0}f,g\big\rangle =
\int_{\R^3}\Big\{
\int_{\R^3}\frac{it}{\pi^2 (|x-y|-t)^2(|x-y|+t-i0)^2} \;\!f(y)\;\!\d y \Big\} \overline{g(x)}\;\!\d x
\end{equation*}

One has thus obtained:

\begin{lemma}
For any
$f \in C_c^\infty(\R^3)$, $g \in \SS$ with $\hat g\in C^\infty_c(\R^3)$ and $\pm t >0$, one has
\begin{equation*}
\big\langle e^{-itH_0}f,g\big\rangle
=
\int_{\R^3}\Big\{
\int_{\R^3}\frac{it}{\pi^2 (|x-y|\pm t)^2(|x-y| \mp t  \pm i 0)^2} \;\!f(y)\;\!\d y \Big\} \overline{g(x)}\;\!\d x,
\end{equation*}
in a formal sense (the precise sense being the one mentioned above).
\end{lemma}


\end{document}